\documentclass[a4paper,USenglish]{lipics-v2021}

\usepackage[showAuthorOrcid]{craigstyle}

\usepackage{hyperref}
\usepackage{babel}
\usepackage{color}
\usepackage{subcaption}
\usepackage{amsmath,bbm}
\usepackage{todonotes}
\usepackage{bussproofs}
\usepackage{wrapfig}

\usepackage{listings} 

\usepackage{xcolor}

\usepackage{tikz}

\usetikzlibrary{calc,fit,positioning}

\allowdisplaybreaks

\newcommand{\defref}[1]{Definition~\ref{#1}}
\newcommand{\lemref}[1]{Lemma~\ref{#1}}
\newcommand{\secref}[1]{Section~\ref{#1}}
\newcommand{\thmref}[1]{Theorem~\ref{#1}}
\newcommand{\figref}[1]{Fig.~\ref{#1}}
\newcommand{\tabref}[1]{Table~\ref{#1}}
\newcommand{\exref}[1]{Example~\ref{#1}}

\presetkeys{todonotes}{disable}{}

\newif\ifoutline
\outlinefalse

\newif\iflong
\longfalse

\newif\ifcolor
\colortrue

\lstdefinestyle{mycode}{
  language=c,
  morekeywords={assert,assume,define,Set,havoc},
  lineskip=.1em,
  numbers=left,
  xleftmargin=2em,
  numberstyle=\scriptsize,
  basicstyle=\small\ttfamily,
  keywordstyle=\color{teal}\ttfamily,
  commentstyle=\color{brown}\ttfamily,
  tabsize=2,
  numbersep=7pt,
  keepspaces=true,
  aboveskip=0pt,
  belowskip=0pt,
  mathescape=true,
  escapeinside={@}{@},
  columns=flexible,
  backgroundcolor=\ifcolor\color{lipicsLightGray}\else\color{white}\fi,
  framerule=\ifcolor 0pt\else 0.4pt\fi
}
\lstdefinestyle{inlinec}{
  style=mycode,
  basicstyle=\ttfamily
}

\def\inl{\lstinline[style=inlinec]}

\newcommand{\contents}[1]{\ifoutline{\footnotesize\color{blue}
    \begin{itemize}
    #1
    \end{itemize}
  \par}\bigskip\fi}

\newcommand{\sig}{\ensuremath{\operatorname{sig}}}

\newcommand{\N}{\mathbbm{N}}
\newcommand{\Q}{\mathbbm{Q}}
\newcommand{\Z}{\mathbbm{Z}}
\newcommand{\R}{\mathbbm{R}}

\newcommand{\naturalSeq}[2]{\ensuremath{#1\;\vdash\;#2}}
\newcommand{\naturalIntSeq}[3]{\ensuremath{#1\;\vdash\;#2\;[\,#3\,]}}
\newcommand{\shortNaturalSeq}[2]{\ensuremath{#1\vdash#2}}
\newcommand{\intSeq}[2]{\ensuremath{#1\; \blacktriangleright\;#2}}
\newcommand{\shortIntSeq}[2]{\ensuremath{#1 \blacktriangleright#2}}
\newcommand{\ruleName}[1]{\textsc{#1}}
\newcommand{\ruleNameInt}[1]{\textsc{#1$^I$}}

\newcommand{\prule}[4]{
  \AxiomC{#1}
  \LeftLabel{#3}
  \ifx#4\empty\else\RightLabel{~(#4)}\fi
  \UnaryInfC{#2}
  \DisplayProof
}

\newcommand{\prulet}[4]{
  \AxiomC{#1}
  \LeftLabel{#3}
  \ifx#4\empty\else\RightLabel{~$\begin{array}{@{}l@{}}
                                   #4
        \end{array}$}\fi
  \UnaryInfC{#2}
  \DisplayProof
}

\newcommand{\prulee}[5]{
  \AxiomC{#1}
  \AxiomC{#2}
  \LeftLabel{#4}
  \ifx#5\empty\else\RightLabel{~(#5)}\fi
  \BinaryInfC{#3}
  \DisplayProof
}

\newcommand{\tr}{\mathit{TR}}

\title{Craig Interpolation in Program Verification}
\author{Philipp R\"ummer}{University of Regensburg, Germany \and
  Uppsala University, Sweden}{philipp.ruemmer@ur.de}{https://orcid.org/0000-0002-2733-7098}{A}
\authorrunning{Philipp R\"ummer}

\begin{document}
\maketitle

\begin{abstract}
  Craig interpolation is used in program verification for automating
  key tasks such as the inference of loop invariants and the
  computation of program abstractions.  This chapter covers some of
  the most important techniques that have been developed in this
  context over the last years, focusing on two aspects: the derivation
  of Craig interpolants modulo the theories and data types used in
  verification and the basic design of verification algorithms
  applying interpolation.

\end{abstract}


\tableofcontents

\section{Introduction}

\contents{
\item What is program verification
\item Challenges in program verification
\item Short history how Craig interpolation found its way to verification
}

In program verification, the question is considered how it can be
shown (effectively or even automatically) whether a given software
program is consistent with its specification. The study of program
verification is as old as computer science and has over the years led
to a plethora of different techniques, studying the question from
various angles and developing approaches that have found widespread
application in both academia and industry.

Craig interpolation found its way to the field of verification when it
was discovered that interpolation can be used to infer
\emph{invariants}~\cite{DBLP:conf/cav/McMillan03}, resulting in a new
paradigm for automating one of the key tasks in verification. The
discovery gave rise to a number of new verification algorithms,
initially focusing on the verification of \emph{hardware circuits,}
for which mainly interpolation in propositional logic was used, and
later extending the methods to the analysis of \emph{software
  programs} and other classes of systems. The newly discovered use of
Craig interpolation initiated the development of more effective
interpolation algorithms and tools for many different logics,
including propositional logic, first-order logic, and the theories
relevant in verification, and was an inspiring factor for the
application of Craig interpolation also in fields other than verification.  Many
automated verification tools today apply Craig interpolation for tasks
such as reasoning about control flow or data, inferring invariants and
other program annotations, or automatically computing program
abstractions.

This chapter focuses on two main aspects related to Craig
interpolation in verification: techniques for computing Craig
interpolants modulo some of the theories used in program verification
and the basic design of verification algorithms applying
interpolation.  Given the huge amount of research in the field, the
chapter does not attempt to provide a complete survey of all results,
but rather introduces some of the main approaches in detail and with
examples and provides a starting point for exploring the world of
interpolation-based verification.

\subsection{Informal Overview: Invariant Inference using Craig Interpolation}
\label{sec:motivating}

\contents{
\item Example of a program with one loop and some assertions in a
  C-like language
\item Craig interpolants correspond to intermediate assertions
\item Craig interpolants are good candidates for inductive loop invariants
\item discuss relationship to Hoare logic?
}

We first illustrate the use of Craig interpolation for inferring loop
invariants using an example. The Fibonacci
sequence~\cite{DBLP:books/aw/GKP1994} is defined by the equations
\begin{equation*}
  F_0 = 0,\quad F_1 = 1,\quad F_{i+2} = F_i + F_{i+1}
\end{equation*}
and is computed by the function shown in \figref{fig:fibonacci}. After
the $i$th iteration of the loop, the program variable~\inl!a!
stores the Fibonacci number~$F_i$ and the variable~\inl!b! the
number~$F_{i+1}$. The function has the precondition~\inl!n >= 0!,
stated as a \inl!requires! clause, and the
postcondition~\inl!\result >= 0! expressing that the result~$F_n$
returned by the function is not negative (``\inl!ensures!'').

The aim of formal verification is to construct a rigorous mathematical
proof that a program is consistent with its specification. In order to
verify the \emph{partial correctness}~\cite{DBLP:books/daglib/0019162}
of \inl!fib!, it has to be proven that the function, whenever it
terminates for a non-negative argument~$n$ with some result~$a$,
ensures the postcondition~$a \geq 0$. Partial correctness proofs are
commonly carried out by an induction over the number of loop
iterations executed by the program. To streamline this reasoning
process, the notion of \emph{inductive loop invariants} (or just
\emph{loop invariants}) has been
introduced~\cite{DBLP:books/daglib/0019162}, which capture properties
that cannot be violated during program execution. More precisely, an
inductive loop invariant for a loop is a formula~$I$ over the program
variables that are in scope in the loop body, satisfying:
\begin{itemize}
\item \emph{Initiation:} whenever the loop is entered, the
  invariant~$I$ holds.
\item \emph{Consecution:} whenever the loop body is executed starting
  in a state in which the invariant~$I$ and the loop condition
  hold, the invariant~$I$ also holds after executing the loop body.
\item \emph{Sufficiency:} whenever the loop terminates in a state
  satisfying the loop invariant~$I$, the rest of the program executes
  without errors and the postcondition is established.
\end{itemize}

\begin{figure}[tb]
  \raisebox{1.5ex}{}

  \begin{minipage}{0.48\linewidth}
  \begin{subfigure}{\linewidth}
\begin{lstlisting}
/**
 * requires n >= 0;
 * ensures  \result >= 0;
 */
int fib(int n) {
  int a = 0;
  int b = 1;
  int t;

  for (int i = 0; i < n; ++i) {
    t = b;
    b = a + b;
    a = t;
  }

  return a;
}
\end{lstlisting}
  
  \caption{Program computing the $n$th element of the
    Fibonacci sequence.}
  \label{fig:fibonacci}
  \end{subfigure}
  \end{minipage}
  \hfill
  \begin{minipage}{0.48\linewidth}
  \begin{subfigure}{\linewidth}
\begin{lstlisting}
assume(n >= 0); // requires
int a = 0; int b = 1; int t;
int i = 0;
assert($I$);
\end{lstlisting}
\caption{Initiation}
\label{fig:fibInit}
  \end{subfigure}

  \bigskip
  \begin{subfigure}{\linewidth}
\begin{lstlisting}
assume($I$);
assume(i < n);  
t = b;
b = a + b;
a = t;
++i;
assert($I$);
\end{lstlisting}
\caption{Consecution}
\label{fig:fibConsec}
  \end{subfigure}

  \bigskip
  \begin{subfigure}{\linewidth}
\begin{lstlisting}
assume($I$);
assume(!(i < n));
assert(a >= 0); // ensures
\end{lstlisting}
\caption{Sufficiency}
\label{fig:fibSuff}
  \end{subfigure}
  \end{minipage}
  
  \caption{Fibonacci program and the conditions for inductive loop
    invariants.}
  \label{fig:motivatingExample}
\end{figure}

For our example program, the three conditions are captured formally on
the right-hand side of \figref{fig:motivatingExample} using the
program instructions~\inl!assert!  and
\inl!assume!~\cite{fla-sax-01-popl}. The semantics of
\inl!assert($\phi$)! is to evaluate a condition~$\phi$ in the current
program state and raise an error if $\phi$ is violated; if $\phi$
holds, \inl!assert($\phi$)! has no effect and program execution
continues. The semantics of \inl!assume($\phi$)! is similar, but
instead of raising an error in case $\phi$ is violated, the statement
will instead suspend the program execution, preventing the rest of the
program from being executed. A statement \inl!assume(false)!, in
particular, stops further program execution and has similar effect as a
non-terminating loop. To deal
with \inl!assert!  and \inl!assume!, we generalize our notion of
partial correctness and say that a program is \emph{partially
  correct} if no \inl!assert!  statement in the program can fail,
regardless of the program inputs and the initial state. In the
example, the sufficient inductive invariants are exactly the
formulas~$I$ for which the three programs in Figs.~\ref{fig:fibInit}
to \ref{fig:fibSuff} are partially correct.

We can observe that all three conditions are satisfied by the formula
$ I = \mathtt{a} \geq 0 \wedge \mathtt{b} \geq 0 $.  This formula~$I$
is therefore a loop invariant and represents a witness for the partial
correctness of \inl!fib!. The simpler formula~$\mathtt{a} \geq 0$,
although it is a property that holds throughout any execution of the
program as well, does not satisfy the \emph{consecution} condition and
is therefore not a correct inductive invariant. More precisely, when
substituting $\mathtt{a} \geq 0$ for $I$ in \figref{fig:fibConsec}, we
obtain a program that is not partially correct, an assertion failure
occurs when the program is run starting in the state~$\mathtt{a} = 0,
\mathtt{b} = -1, \mathtt{i} = 0, \mathtt{n} = 1$.

\begin{figure}[tb]
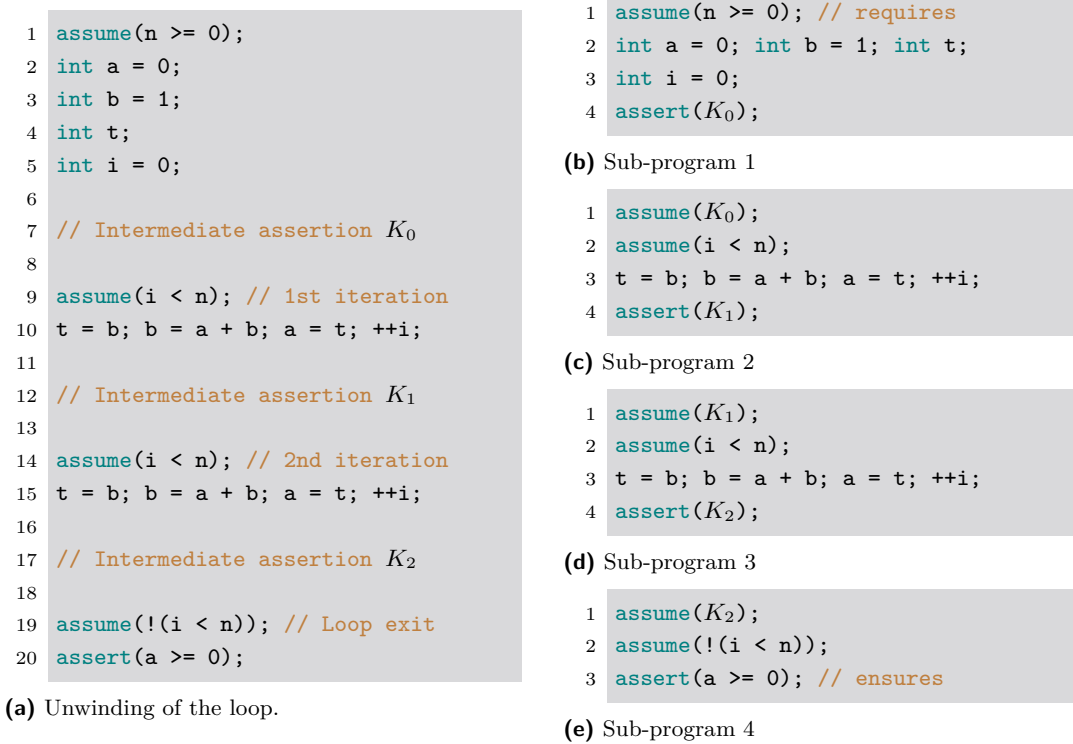

  \raisebox{1.5ex}{}
  
  \begin{minipage}{0.48\linewidth}
  \begin{subfigure}{\linewidth}
\begin{lstlisting}
assume(n >= 0);
int a = 0;
int b = 1;
int t;
int i = 0;

// Intermediate assertion $K_0$

assume(i < n); // 1st iteration
t = b; b = a + b; a = t; ++i;

// Intermediate assertion $K_1$

assume(i < n); // 2nd iteration
t = b; b = a + b; a = t; ++i;

// Intermediate assertion $K_2$

assume(!(i < n)); // Loop exit
assert(a >= 0);
\end{lstlisting}
\caption{Unwinding of the loop.}
  \label{fig:fibonacciUnwinding}
\end{subfigure}
\end{minipage}
  \hfill
  \begin{minipage}{0.48\linewidth}
  \begin{subfigure}{\linewidth}
\begin{lstlisting}
assume(n >= 0); // requires
int a = 0; int b = 1; int t;
int i = 0;
assert($K_0$);
\end{lstlisting}
\caption{Sub-program 1}
\label{fig:fibUnwind1}
  \end{subfigure}

  \bigskip
  \begin{subfigure}{\linewidth}
\begin{lstlisting}
assume($K_0$);
assume(i < n);  
t = b; b = a + b; a = t; ++i;
assert($K_1$);
\end{lstlisting}
\caption{Sub-program 2}
\label{fig:fibUnwind2}
  \end{subfigure}

  \bigskip
  \begin{subfigure}{\linewidth}
\begin{lstlisting}
assume($K_1$);
assume(i < n);  
t = b; b = a + b; a = t; ++i;
assert($K_2$);
\end{lstlisting}
\caption{Sub-program 3}
\label{fig:fibUnwind3}
  \end{subfigure}

  \bigskip
  \begin{subfigure}{\linewidth}
\begin{lstlisting}
assume($K_2$);
assume(!(i < n));
assert(a >= 0); // ensures
\end{lstlisting}
\caption{Sub-program 4}
\label{fig:fibUnwind4}
  \end{subfigure}
  \end{minipage}

  \caption{Unwinding two iterations of the loop in the Fibonacci program.}
  \label{fig:motivatingExample2}
 \end{figure}

\bigskip

Coming up with inductive loop invariants, like the formula
$\mathtt{a} \geq 0 \wedge \mathtt{b} \geq 0$, is a known challenge in
verification, for the same reasons as it is not easy to find induction
formulas in proofs. Craig interpolation is a useful and
widely applied technique for finding invariant candidates. To identify
the connection between Craig interpolation and loop invariants,
consider the unwinding of the Fibonacci program in
\figref{fig:fibonacciUnwinding}, which is obtained by listing the
statements on the execution path that executes the loop body
twice. The assumptions added to the program capture the status of the
loop condition, which is $\mathit{true}$ for two loop iterations and
$\mathit{false}$ after that, and ensure that the executions of the
program in \figref{fig:fibonacciUnwinding} correspond exactly to the
executions of the original program in \figref{fig:fibonacci} with two
loop iterations. Unwinding loops to obtain loop-free programs is a
common technique in verification.

Consider then possible \emph{intermediate assertions}~$K_0, K_1, K_2$
that enable us to decompose the program in
\figref{fig:fibonacciUnwinding} into multiple parts at the marked
locations. We say that a formula~$A$ is an intermediate assertion for
a partially correct program~\inl!$P$;$Q$! if the programs
\inl!$P$;assert($A$)! and
\inl!assume($A$);$Q$! are both partially correct as well; in
other words, if $A$ is a postcondition that is ensured by the first
sub-program~$P$, but that is also precondition ensuring error-free
execution of the second sub-program~$Q$. The three intermediate
assertions in \figref{fig:fibonacciUnwinding} similarly split the
program into four sub-programs, shown on the right-hand side of
\figref{fig:motivatingExample2}, each of which needs to be partially
correct for the intermediate assertions to be suitable.

\bigskip

The use of Craig interpolants for finding inductive loop invariants is
based on two key observations:
\begin{enumerate}[(i)]
\item intermediate assertions are a relaxation of inductive loop
  invariants, and
\item intermediate assertions are a case of Craig interpolants.
\end{enumerate}
To see point~(i), we can compare the conditions in
\figref{fig:motivatingExample} with the conditions given in
\figref{fig:motivatingExample2}. Conditions~\ref{fig:fibInit} and
\ref{fig:fibUnwind1} correspond to each other, as do
conditions~\ref{fig:fibSuff} and \ref{fig:fibUnwind4}, whereas the
consecution condition~\ref{fig:fibConsec} is only approximated by
conditions~\ref{fig:fibUnwind2} and \ref{fig:fibUnwind3}. Therefore,
any formula~$I$ satisfying the conditions in
\figref{fig:motivatingExample}, when substituted for $K_0, K_1, K_2$,
will also satisfy the conditions in \figref{fig:motivatingExample2} and
give rise to a correct decomposition of the program. Vice versa, not
all possible intermediate assertions~$K_0, K_1, K_2$ are correct loop
invariants, since intermediate assertions will in general not satisfy
consecution~\ref{fig:fibConsec}. Given a method to explore possible
intermediate assertions~$K_0, K_1, K_2$, however, we can generate
\emph{candidates} for inductive loop invariants. We survey some of the
algorithms that have been proposed for generating and refining
invariant candidates in this way in \secref{sec:verification}.

\begin{figure}[tb]
  \begin{minipage}{0.48\linewidth}
    \bigskip
  \begin{subfigure}{\linewidth}
\begin{lstlisting}
assume(n >= 0);
int a = 0, b = 1;
int t, i = 0;
assume(i < n);
t = b; b = a + b; a = t; ++i;

// Intermediate assertion $K$

assume(i < n);
t = b; b = a + b; a = t; ++i;
assume(!(i < n));
assert(a >= 0);
\end{lstlisting}
\caption{Computation of a single intermediate assertion.}
  \label{fig:oneIntermediateAssertion}
\end{subfigure}
\end{minipage}
  \hfill
  \begin{minipage}{0.48\linewidth}
  \begin{subfigure}{\linewidth}
    \begin{equation*}
        \begin{array}{@{}ll@{}}
          n_0 \geq 0 \wedge\mbox{}\\
          a_0 = 0 \wedge b_0 = 1 \wedge i_0 = 0 \wedge\mbox{}\\
          i_0 < n_0 \wedge\mbox{} & \smash{\left.\rule{0mm}{8ex}\right\}} A\\
          t_1 = b_0 \wedge b_1 = a_0 + b_0 \wedge a_1 = t_1 \wedge\mbox{}\\
          i_1 = i_0 + 1\\[0.3ex]
           \wedge\mbox{}\\[0.3ex]
          i_1 < n_0  \wedge\mbox{}\\
          t_2 = b_1 \wedge b_2 = a_1 + b_1 \wedge a_2 = t_2 \wedge\mbox{}\\
           i_2 = i_1 + 1 \wedge\mbox{}  & \smash{\left.\rule{0mm}{8ex}\right\}} C\\
           i_2 \not< n_0\\[0.5ex]
           \to a_2 \geq 0
        \end{array}
    \end{equation*}
\caption{Verification condition for partial correctness.}
  \label{fig:fibVC}
\end{subfigure}
  \end{minipage}

  \caption{Intermediate assertions.}
  \label{fig:motivatingExample3}
 \end{figure}

To see point~(ii), we focus on the case of just a single intermediate
assertion to be computed for the unwound Fibonacci program, shown in
\figref{fig:oneIntermediateAssertion}. The partial correctness of the
program can be checked by translating the program, including the
\inl!assert! and \inl!assume! statements, to a
\emph{verification condition}~\cite{DBLP:books/daglib/0019162}, which
is in this case a formula over linear integer arithmetic. The
verification condition is shown in \figref{fig:fibVC} and is constructed
in such a way that the condition is \emph{valid} if and only if the
program is \emph{partially correct.}  For this, the condition lists
all assumptions made in the program, includes the assignments to
variables in the form of equations, and states the assertion as the
last inequality~$a_2 \geq 0$. Since the program variables can change
their values during program execution, a program variable can
correspond to multiple symbols in the verification condition: for
instance, program variable~\inl!a! turns into the logical
constants~$a_0, a_1, a_2$.

By comparing the verification condition in \figref{fig:fibVC} with our
definition of intermediate assertions, we can see that the correct
intermediate assertions~$K$ correspond to the formulas~$\phi$ for
which the implications~$A \to \phi$ and $\phi \to C$ are valid, with
$A, C$ as marked in \figref{fig:fibVC}. Indeed, $K$ has to be a
formula that is entailed by the program statements prior to the
respective point in the program, which is expressed by the implication
$A \to \phi$, and it has to be strong enough to ensure that the rest
of the program executes correctly, hence $\phi \to C$. This already
indicates that the intermediate assertions correspond to the
\emph{Craig interpolants} for the implication~$A \to C$, as
interpolants have to satisfy the same implication conditions.

To bridge the gap between intermediate assertions and Craig
interpolants further, we have to consider which symbols are allowed to
occur in the intermediate assertion~$K$. The formula~$K$ is
supposed to be an expression over the program variables
\inl!n, a, b, t, i!.
In the verification condition, those variables have been replaced with
different indexed versions; at the point where $K$ is supposed to be
evaluated (line~7 of \figref{fig:oneIntermediateAssertion}), the most
recent versions are named $n_0, a_1, b_1, t_1, i_1$. Naturally, those
are exactly the symbols that occur in both $A$ and $C$, as
all the older previous versions of the variables will be present only
in $A$, while future versions can only occur in $C$.  Craig
interpolants for the implication~$A \to C$ therefore range over
exactly the right set~$n_0, a_1, b_1, t_1, i_1$ of symbols. To derive
intermediate assertions for our program, we can first compute a Craig
interpolant~$\phi$ for $A \to C$, and then rename the 
symbols to obtain the intermediate
assertion~$K = \phi[n_0 \mapsto \mathtt{n}, a_1 \mapsto \mathtt{a}, b_1 \mapsto \mathtt{b},
t_1 \mapsto \mathtt{t}, i_1 \mapsto \mathtt{i}]$.


We will return to the Fibonacci example later in this chapter, when it
is discussed how we can compute Craig interpolants modulo the theories
relevant for the program. One possible interpolant we will then derive
is the formula~$\phi = (b_1 \geq 1)$, which can be translated to the
intermediate assertion~$K = (\mathtt{b} \geq 1)$.

\subsection{Organization of the Chapter}

The rest of the chapter consists of two main parts. In
\secref{sec:satIntMT}, we introduce the required background in
logic and formally define several notions of Craig interpolation. We
then introduce some of the techniques that have been developed in
the verification field for computing Craig interpolants, focusing on
interpolants for arithmetic formulas that are essential for
verification. In \secref{sec:verification}, the second part, we 
put the interpolation procedures to use by exploring several
interpolation-based verification algorithms.


\section{Satisfiability and Interpolation Modulo Theories}
\label{sec:satIntMT}

\subsection{Basic Definitions}

\contents{
\item First-order logic
\item Theories
\item Interpolation in theories, quantifier-free interpolation
\item Reverse interpolants (connect to foundation chapter, ``separators'')
\item Sequence interpolants
\item Tree interpolants (extension of ``parallel interpolants'')
\item Combination of propositional and theory interpolation
\item Support for interpolation in SMT solvers
}


Automatic approaches in
verification often work in a setting of \emph{quantifier-free
  theories} in \emph{first-order logic,} where the theories are chosen
to reflect the data types found in programming languages, while the
restriction to quantifier-free formulas is motivated by the desire to
stay in a decidable language. For reasoning about quantifier-free
formulas modulo theories, the \emph{Satisfiability Modulo Theories}
(SMT) paradigm~\cite{BSST21} has been developed, in which SAT solvers
are combined with theory solvers to obtain efficient decision
procedures.  Similarly, much of the research on Craig interpolation in
verification has focused on methods to compute quantifier-free Craig
interpolants modulo theories.

\subsubsection{First-Order Logic and Theories}

We assume a setting of multi-sorted (classical) first-order logic, see
for instance \cite{DBLP:books/daglib/0022394}. Slightly extending the
notation introduced in \refchapter{chapter:predicate}, a
signature~$\Sigma$ is now a collection of relation symbols, constant
and function symbols, as well as sorts, where each predicate~$R$ is
endowed with a uniquely defined
rank~$\alpha(R) = (\sigma_1, \ldots, \sigma_n)$ of argument sorts, and
each constant and function symbol~$f$ with a
rank~$\alpha(f) = (\sigma_1, \ldots, \sigma_n, \sigma_0)$ consisting
of the argument sorts~$\sigma_1, \ldots, \sigma_n$ and the result
sort~$\sigma_0$.  Formulas~$\phi$ and terms~$t$ are defined by the
following grammar:
\begin{equation*}
  \phi ~::=~ R(t, \ldots, t) \mid t = t \mid \top \mid \phi \wedge \phi
  \mid \neg \phi \mid \exists x : \sigma.\, \phi~,
  \qquad\qquad
  t ~::=~ x \mid  f(t, \ldots, t)~.
\end{equation*}
In the grammar, $R$ ranges over relation symbols, $f$ over function
symbols, $x$ over variables, and $\sigma$ over sorts. Constants are
identified with nullary functions. We generally assume that terms and
formulas are well-typed, which means that relation and function
symbols are only applied to terms of the right sort, and that $=$ is
only used to compare terms of the same sort.
We also use the short-hand notations~$\bot$, $\phi \vee \phi$,
$\phi \to \phi$, and $\forall x : \sigma.\, \phi$, all with the usual
definition. Given a formula~$\phi$ or term~$t$, we denote the sets of
occurring relation and function symbols by
$\sig(\phi) \subseteq \Sigma$ and $\sig(t) \subseteq \Sigma$,
respectively.

A \emph{closed formula} or \emph{sentence} is a formula in which all
variables are bound. A \emph{ground formula} is a formula that does
not contain any variables (and, therefore, no quantifiers). Note that
ground formulas can still contain constant symbols. A
\emph{quantifier-free} formula is a formula that does not contain
quantifiers, but that might contain variables.

The semantics of formulas is defined, as is usual, in terms of
\emph{$\Sigma$-structures}~$S$ that map each sort~$\sigma$ in $\Sigma$
to some non-empty domain~$D_\sigma$, each relation symbol in $\Sigma$
to a set-theoretical relation over the argument domains, and each
function symbol in $\Sigma$ to a set-theoretical function from the
argument domains to the result domain. To define the satisfiability of
a formula modulo theories, we need the notion of
\emph{extensions} of signatures and structures. We say that a
signature~$\Omega$ is an \emph{extension} of $\Sigma$, denoted
$\Sigma \subseteq \Omega$, if $\Omega$ contains more sorts, relation
symbols, or function functions than $\Sigma$, without changing the
rank~$\alpha(R)$ or $\alpha(f)$ of symbols that already exist in
$\Sigma$. If $\Sigma \subseteq \Omega$, then a
$\Sigma$-structure~$S_\Sigma$ is \emph{extended} by an
$\Omega$-structure~$S_\Omega$ if $S_\Omega$ interprets the sorts and
symbols already present in $\Sigma$ in the same way as $S_\Sigma$,
while adding interpretations for the new sorts and symbols in
$\Omega$.

\begin{definition}[Theory]
  \label{def:theory}
  A \emph{theory}~$T$ is defined by a signature~$\Sigma_T$ and a
  set~$\mathcal{S}_T$ of structures over $\Sigma_T$. A sentence~$\phi$
  over the signature~$\Sigma_T$ is called \emph{$T$-satisfiable (or
    $T$-sat)} if it evaluates to $\mathit{true}$ for some structure in
  $\mathcal{S}_T$. A sentence~$\phi$ over an
  extension~$\Omega \supseteq \Sigma_T$ is called
  \emph{$T$-satisfiable} if it evaluates to $\mathit{true}$ in some
  extension of a structure in $\mathcal{S}_T$ to $\Omega$.
\end{definition}

For sentences in some signature~$\Omega$ that is larger than
$\Sigma_T$, this means that additional sorts and relation or function
symbols can be interpreted arbitrarily.
As a special case, we say that a sentence~$\phi$ is simply
\emph{satisfiable} if it is $T_\emptyset$-satisfiable in the empty
theory~$T_\emptyset$.
Sentences that are not $T$-satisfiable are called
\emph{$T$-unsatisfiable (or $T$-unsat),} and sentences whose negation
is $T$-unsatisfiable are called \emph{$T$-valid}. We also speak of a
set~$M$ of sentences being $T$-satisfiable if there is a common
structure satisfying all sentences in $M$.

We write $\Gamma \models_T \psi$ if the set~$\Gamma$ of sentences
entails the sentence~$\psi$ modulo $T$: whenever an extension of a
structure in $\mathcal{S}_T$ satisfies all sentences in $\Gamma$, it
also satisfies $\psi$. As a special case, we write
$\phi \models_T \psi$ if sentence~$\phi$ entails sentence~$\psi$
modulo $T$, which is equivalent to $\{\phi\} \models_T \psi$ and to
the conjunction~$\phi \wedge \neg \psi$ being $T$-unsatisfiable. We
write~$\phi \equiv_T \psi$ for two sentences being $T$-equivalent,
defined as $\phi \models_T \psi$ and $\psi \models_T \phi$, and leave
out the $T$ when it is clear from the context.

\subsubsection{Craig Interpolation Modulo Theories}

Since algorithms in verification (just like in automated reasoning)
are usually formulated in terms of \emph{satisfiable} or
\emph{unsatisfiable} formulas, instead of \emph{valid} or
\emph{counter-satisfiable} formulas, it is common in the area to also
define Craig interpolation in the ``reversed'' way. The
relationship between standard Craig interpolants and \emph{separators}
has already been discussed in
\refchapter{chapter:propositional}, clarifying that the
difference between the two notions is mostly superficial, at least in
classical logic. In order to avoid confusion with the definition in the
literature, we mostly stick to this reversed version of interpolation
throughout the chapter, but also introduce the standard version for
completeness reasons.

\begin{definition}[Craig interpolant modulo theory]
  Suppose that $T$ is a theory and $A \to C$ a sentence. A
  sentence~$I$ is called a \emph{Craig interpolant modulo $T$} for
  $A \to C$ if the following conditions hold:
  \begin{itemize}
  \item $A \models_T I$ and $I \models_T C$, and
  \item $\sig(I) \subseteq (\sig(A) \cap \sig(C)) \cup \Sigma_T$.
  \end{itemize}
\end{definition}

Compared to the definition of Craig interpolants in first-order logic
from \refchapter{chapter:predicate}, we now consider
entailment~$\models_T$ modulo the theory~$T$, and we allow the
interpolant~$I$ to include arbitrary \emph{interpreted} symbols from
the theory signature~$\Sigma_T$. As a further, more technical
difference, we only consider closed formulas (sentences) in our
definition; this restriction could be lifted easily but is usually
unproblematic, since free variables can be replaced with constant
symbols. The definition in \refchapter{chapter:predicate}
otherwise corresponds to our definition for the special case of the
empty theory~$T_\emptyset$.

\begin{example}
  Consider the implication~$A \to C$, with $A, C$ defined as in
  \figref{fig:fibVC} over the theory of linear integer arithmetic
  (LIA). Correct Craig interpolants modulo LIA include
  $I_1 = b_1 \geq 1$ and $I_2 = a_1 \geq 0 \wedge b_1 \geq 0$.
  \lipicsEnd
\end{example}

As discussed in \refchapter{chapter:propositional}, an
interpolation problem $A \to C$ can equivalently be expressed as a
conjunction, leading to the alternative reverse notion of
interpolation:

\begin{definition}[Separator, reverse Craig interpolant]
  \label{def:separator}
  Suppose that $T$ is a theory and $A \wedge B$ a sentence. A
  sentence~$I$ is called a \emph{$T$-separator} or \emph{reverse Craig
    interpolant modulo $T$} for $A \wedge B$ if the following
  conditions hold:
  \begin{itemize}
  \item $A \models_T I$ and $I \wedge B \models_T \bot$, and
  \item $\sig(I) \subseteq (\sig(A) \cap \sig(B)) \cup \Sigma_T$.
  \end{itemize}
\end{definition}

\begin{theorem}[Existence of $T$-separators~\cite{DBLP:conf/cade/KovacsV09}]
  \label{thm:separators}
  Suppose that $T$ is a theory and $A \wedge B$ a sentence that is
  $T$-unsatisfiable. Then there is a $T$-separator for $A \wedge B$;
  more precisely, there are $T$-separators~$I_1, I_2$ such that:
  \begin{itemize}
  \item $A \models_T I_1$ and $I_1 \wedge B \models \bot$,
  \item $A \models I_2$ and $I_2 \wedge B \models_T \bot$,
  \item $\sig(I_1) \subseteq (\sig(A) \cup \Sigma_T) \cap \sig(B)$, and
  \item $\sig(I_2) \subseteq \sig(A) \cap (\sig(B) \cup \Sigma_T)$.
  \end{itemize}
\end{theorem}

The theorem provides a stronger result about the existence of
separators than required by \defref{def:separator}. The
sentences~$I_1, I_2$ in the theorem are both $T$-separators according
to \defref{def:separator}, but we are even at liberty to assume that
one of the two entailment conditions is \emph{general}
entailment~$\models$, and only the other entailment is modulo the
theory~$T$. Similarly, the theorem provides stronger guarantees about
the signature used in $I_1, I_2$, as the formulas only have to refer
to interpreted symbols from $\Sigma_T$ that also occur in one of the
two formulas~$A, B$.

Although it provides a general guarantee for the existence of reverse
interpolants, \thmref{thm:separators} is usually not strong enough for
applications in verification, as the theorem does not make any
statement about the occurrence of \emph{quantifiers} in $I_1$ and $I_2$. In
general, even if the formulas~$A, B$ are quantifier-free, there might
only be $T$-separators containing quantifiers, which is usually less
desirable for applications.\todo{amalgamation}

\begin{example}
  \label{ex:arrays}
  McCarthy proposed a theory of functional arrays that is today widely
  used in program verification~\cite{mccarthy:verif}. The theory
  provides functions~$\mathit{select}(M, a)$ for reading from an
  array~$M$ at index~$a$, and $\mathit{store}(M, a, x)$ for updating
  an array~$M$ at index~$a$ to the new value~$x$. In the theory of
  arrays, the conjunction
  \begin{equation*}
    \underbrace{
      M' = \mathit{store}(M, a, x)}_A \wedge
    \underbrace{
      b \not= c \wedge
      \mathit{select}(M, b) \not= \mathit{select}(M', b) \wedge
      \mathit{select}(M, c) \not= \mathit{select}(M', c)}_B
  \end{equation*}
  is unsatisfiable, but all reverse interpolants contain
  quantifiers~\cite{DBLP:journals/tcs/McMillan05}. Indeed,
  interpolants can only refer to the symbols~$M', M$, which occur in
  both $A$ and $B$, as well as the array functions~$\mathit{select}$
  and $\mathit{store}$. A quantified interpolant is
  $\exists a, x.\, M' = \mathit{store}(M, a, x)$.
  \lipicsEnd
\end{example}

\begin{definition}[General quantifier-free
  interpolation~\cite{DBLP:journals/tocl/BruttomessoGR14}]
  A theory~$T$ admits \emph{general quantifier-free interpolation} if
  for every $T$-unsatisfiable conjunction~$A \wedge B$ there is a
  quantifier-free $T$-separator.
\end{definition}

As the example illustrates, the standard theory of arrays does not
admit general quantifier-free interpolation, while many other theories
(for instance, Linear Real Arithmetic, LRA, \secref{sec:lra}) have
this property. It is possible to extend the theory of arrays, however,
to obtain quantifier-free
interpolation~\cite{DBLP:conf/rta/BruttomessoGR11,DBLP:conf/popl/TotlaW13,DBLP:conf/cade/HoenickeS18}
(also see \secref{sec:proofBased}).

It is often useful to consider a somewhat weaker criterion for
quantifier-free interpolation:

\begin{definition}[Plain quantifier-free
  interpolation~\cite{DBLP:journals/tocl/BruttomessoGR14}]
  A theory~$T$ admits \emph{plain quantifier-free interpolation} if
  for every $T$-unsatisfiable conjunction~$A \wedge B$ that only
  contains symbols from $\Sigma_T$ and constant symbols there is a
  quantifier-free $T$-separator.
\end{definition}

The difference between ``general'' and ``plain'' is that the former
demands quantifier-free interpolants even for
conjunctions~$A \wedge B$ that contain relation or function symbols
that are not part of $T$, whereas the latter is happy with only
additional constant symbols.

\begin{example}
  \label{ex:lia-general-qf-interpolation}
  The theory of linear integer arithmetic (LIA) admits plain
  quantifier-free interpolants, provided that the theory includes
  divisibility statements~$k \mid \cdot$ for every $k \in \N$; a
  formula~$k \mid t$ is read ``$k$ divides $t$'' and asserts that $t$
  is a multiple of $k$. Plain quantifier-free interpolation
  follows thanks to the observation that weakest and strongest interpolants in LIA
  can be computed using quantifier
  elimination~\cite{DBLP:books/daglib/0022394}, in exactly the same
  way as was shown for propositional logic in
  \refchapter{chapter:propositional}. Whether LIA admits
  \emph{general quantifier-free interpolation} depends on the exact
  signature of the theory that is assumed. To see this, consider the
  unsatisfiable conjunction
  \begin{equation*}
    \underbrace{
      x = 2\cdot y \wedge p(y)}_A \wedge
    \underbrace{
      x = 2\cdot z \wedge \neg p(z)}_B
  \end{equation*}
  where $p$ is a relation symbol that is not in the signature of LIA.
  A separator for the conjunction~$A \wedge B$ is the
  formula~$I = p(x \div 2)$, where $\div$ denotes division with
  remainder. If the function~$\div$ is not included in the signature
  of LIA, no quantifier-free separators
  exist~\cite{DBLP:conf/vmcai/BrilloutKRW11}, and the theory loses the
  general quantifier-free interpolation property.
  \lipicsEnd
\end{example}

\subsubsection{Extended Forms of Craig Interpolation}
\label{sec:extendedCraig}

The notion of separators (or reverse Craig interpolants) defined in
the previous section is sometimes called \emph{binary interpolation,}
as only two formulas~$A, B$ are involved. It is possible to define
corresponding $n$-ary versions, in which an arbitrary number of
formulas can be separated, which is often advantageous in
verification. The two main extended notions of interpolation that are
used today, and that are commonly supported by SMT solvers, are
\emph{sequence} and \emph{tree} interpolation.

\begin{definition}[Sequence
  interpolant~\cite{DBLP:conf/popl/HenzingerJMM04,DBLP:conf/cav/McMillan06}]
  \label{def:sequenceInterpolants}
  Suppose $T$ is a theory and $A_1 \wedge \cdots \wedge A_n$ a
  sentence. A \emph{sequence interpolant modulo $T$} or
  \emph{inductive sequence of interpolants modulo $T$} for
  $A_1 \wedge \cdots \wedge A_n$ is a sequence~$I_0, I_1, \ldots, I_n$
  of sentences such that:
  \begin{itemize}
  \item $I_0 = \top, I_n = \bot$,
  \item for all $i \in \{1, \ldots, n\}$:
    $I_{i-1} \wedge A_i \models_T I_i$, and
  \item for all $i \in \{0, \ldots, n\}$:
    $\sig(I_i) \subseteq
    (\sig(A_1 \wedge \cdots \wedge A_i) \cap \sig(A_{i+1} \wedge \cdots \wedge
    A_n)) \cup \Sigma_T$.
  \end{itemize}
\end{definition}

It is easy to see that a separator~$I$ for $A_1 \wedge A_2$
corresponds to the sequence interpolant~$\top, I, \bot$, so that
sequence interpolation strictly extends binary interpolation. It
should be noted, as well, that the individual formulas of a sequence
interpolant depend on each other, as each formula~$I_{i-1}$ has to be
strong enough to imply, together with $A_i$, the next
formula~$I_i$. Sequence interpolants can be derived by repeatedly
computing binary interpolants: we set $I_0 = \top$ and then compute
each formula~$I_i$, for $i = \{1, \ldots, n\}$, by computing a
separator for the
conjunction~$(I_{i-1} \wedge A_i) \wedge (A_{i+1} \wedge \cdots \wedge
A_n)$. For reasons of efficiency, this iterative computation of
sequence interpolants is rarely used, however, SMT solvers instead
implement interpolation procedures that can compute the whole sequence
interpolant in one go.

\begin{example}
  As we observed in \secref{sec:motivating} and \figref{fig:fibVC}, a
  single intermediate assertion for a program can be derived using a
  binary interpolant by translating the statements on a program path
  to an implication~$A \to C$. We can use sequence interpolants to
  derive multiple consecutive assertions in a program. For this,
  consider again the version of the Fibonacci example shown in
  \figref{fig:motivatingExample2}, in which the loop was unwound two
  times with the goal of finding intermediate
  assertions~$K_0, K_1, K_2$. The locations of the intermediate
  assertions in the program indicate that we should now group together
  the program statements to obtain four conjuncts; the
  assertion~\inl!a >= 0! is negated to state that we are
  interested in runs violating the assertion:
  \begin{equation*}
    \underbrace{\raisebox{-4ex}{}
      \begin{array}{@{}l@{}}
      n_0 \geq 0 \wedge a_0 = 0 \wedge\mbox{}\\
      b_0 = 1 \wedge i_0 = 0
    \end{array}}_{A_1}
  ~\wedge~
    \underbrace{\raisebox{-4ex}{}
    \begin{array}{@{}l@{}}
      i_0 < n_0 \wedge
      t_1 = b_0 \wedge\mbox{}\\
      b_1 = a_0 + b_0 \wedge\mbox{}\\ a_1 = t_1 \wedge
      i_1 = i_0 + 1
    \end{array}}_{A_2}
  ~\wedge~
    \underbrace{\raisebox{-4ex}{}
     \begin{array}{@{}l@{}}
      i_1 < n_0  \wedge
      t_2 = b_1 \wedge\mbox{}\\
      b_2 = a_1 + b_1 \wedge\mbox{}\\ a_2 = t_2 \wedge
      i_2 = i_1 + 1
    \end{array}}_{A_3}
  ~\wedge~
    \underbrace{\raisebox{-4ex}{}
     \begin{array}{@{}l@{}}
       i_2 \not< n_0 \wedge
        a_2 \not\geq 0
    \end{array}}_{A_4}
  \end{equation*}
  The conjunction~$A_1 \wedge A_2 \wedge A_3 \wedge A_4$ is
  unsatisfiable, and a sequence interpolant is given by the formulas
  \begin{equation*}
    \top,\; a_0 + b_0 \geq 1,\; b_1 \geq 1,\; a_2 \geq 1,\; \bot~.
  \end{equation*}
  The formulas in the sequence interpolant can be translated to
  intermediate program assertions:
  $
    K_0 = \text{\inl!a + b >= 1!},\;
    K_1 = \text{\inl!b >= 1!},\;
    K_2 = \text{\inl!a >= 1!}
    $.
  \lipicsEnd
\end{example}


Tree interpolants generalize inductive sequences of interpolants, and
are designed with inter-procedural program verification in mind. In
this context, the tree structure of the interpolation problem
corresponds to the call graph of a program. Tree interpolants make it
possible to derive not only intermediate assertions in programs, but
also annotations that relate multiple program states, for instance
postconditions~\cite{DBLP:conf/popl/HeizmannHP10,DBLP:conf/aplas/GuptaPR11}
(also see \secref{sec:recursive}).

\begin{definition}[Tree interpolant~\cite{DBLP:conf/popl/HeizmannHP10,DBLP:journals/fmsd/RummerHK15,DBLP:conf/atva/GurfinkelRS13}]
  Suppose $T$ is a theory and $G = (V, E)$ a finite rooted tree,
  writing $E(v, w)$ to express that the node~$w$ is a direct child of
  $v$. Further, suppose that each node~$n \in V$ of the tree is
  labeled with a sentence~$A(n)$. A function~$I$ that labels every
  node~$n \in V$ with a sentence~$I(n)$ is called a \emph{tree
    interpolant modulo $T$} (for $G$ and $A$) if the following
  properties hold:
  \subdefinition{treeInt1}
 for the root node~$v_0 \in V$, it is the case that
 $I(v_0) = \bot$,
 \subdefinition{treeInt2}
 for every node~$v \in V$, the following entailment holds:
 \begin{equation*}
   \qquad\quad
    \bigwedge_{\substack{w \in V \\ (v, w) \in E}} I(w)
    \wedge A(v)
    ~\models_T~
    I(v)~,
  \end{equation*}
 \subdefinition{treeInt3}
 for every node~$v \in V$, it is the case that:
 \begin{equation*}
   \qquad\quad
    \sig(I(v)) ~\subseteq~
    \Big(
    \bigcup_{\substack{w \in V\\E^*(v, w)}} \sig(A(w)) ~~\cap
    \bigcup_{\substack{w \in V\\\neg E^*(v, w)}} \sig(A(w))
    \Big) \cup
    \Sigma_T
  \end{equation*}
  \hspace{3em} 
  where $E^*$ is the reflexive transitive closure of $E$.
\end{definition}

Like in the case of sequence interpolants, also tree interpolants can
be derived by repeatedly computing binary
interpolants~\cite{DBLP:journals/fmsd/RummerHK15}, although SMT
solvers usually provide native support for tree interpolants as well.

Tree interpolation generalizes all of the more specialized forms of
interpolation we have encountered so far. To see that tree
interpolation subsumes sequence interpolation, note that an $n$-ary
conjunction~$A_1 \wedge \cdots \wedge A_n$ can be translated to a
tree~$G_1 = (\{1, \ldots, n\}, \{(i, i-1) \mid 2 \leq i \leq n\})$ in
which each node~$i$ is the child of node~$i+1$ and is labeled with the
formula~$A_i$. The root of the tree is the node~$n$. Similarly, we can
reduce \emph{parallel
interpolants}~\cite{Kourousias_Makinson_2007,DBLP:journals/synthese/Parikh11},
as introduced in \refchapter{chapter:propositional}, to tree
interpolants by arranging an
implication~$\varphi_1 \wedge \cdots \wedge \varphi_n \to \psi$ as a
tree~$G_2 = (\{0, 1, \ldots, n\}, \{(0, i) \mid 1 \leq i \leq n\})$ in
which the root~$0$ is labeled with $\neg\psi$ and each child~$i$ with
the formula~$\varphi_i$. Parallel interpolants are related to the
notion of \emph{symmetric interpolants}~\cite{DBLP:journals/lmcs/JhalaM07},
which can be similarly reduced to tree interpolation.

Tree interpolants can be generalized further to \emph{disjunctive
  interpolants}~\cite{DBLP:conf/cav/RummerHK13}, which have, however,
not found wider adoption in verification. It has also been observed
that the different variants of interpolation correspond to solving
particular recursion-free fragments of constrained Horn
clauses~\cite{DBLP:journals/fmsd/RummerHK15}, which is a connection
that we return to in \secref{sec:chc}.


\subsection{Interpolation Modulo Linear Arithmetic}

\contents{
\item Linear rational arithmetic (LRA)
\item Linear integer arithmetic (LIA)
\item Discussion of feasible interpolation for LIA
\item top-down vs. bottom-up
\item (connect to chapter on proof theory)
\item Strict inequalities
\item Other examples: arrays
}

We now survey some of the techniques that have been developed for
constructing quantifier-free interpolants modulo a theory~$T$,
focusing on the case of \emph{linear arithmetic} that is ubiquitous in
verification. For arithmetic formulas, interpolants are commonly
\emph{extracted from proofs} in some suitable proof system, resembling
the approaches that were discussed in earlier chapters. For
verification, extracting interpolants from proofs is attractive, since
proofs tend to focus on the relevant information in a verification
condition. When computing intermediate assertions (as discussed in
\secref{sec:motivating}) using proof-based interpolation, only the
program variables will be mentioned that were important for the SMT
solver to prove the correctness of the program. Interpolants from
proofs are therefore a useful tool when computing program
abstractions.


\subsubsection{Linear Real Arithmetic (LRA)}
\label{sec:lra}

One of the most common theories supported by interpolating SMT solvers
is \emph{linear arithmetic} over real numbers or the
integers~\cite{DBLP:books/daglib/0090562,DBLP:conf/cav/DutertreM06}. In
the theory of \emph{Linear Real Arithmetic,} LRA, we consider formulas
that contain linear
inequalities~$0 \leq \alpha_1 x_1 + \cdots + \alpha_k x_k + \alpha_0$,
where $\alpha_0, \alpha_1, \ldots, \alpha_k \in \Q$ are rational
numbers and $x_1, \ldots, x_k$ are constant symbols of
sort~$\mathit{Real}$. The semantics of scaling with some rational
number~$\alpha x$, addition~$+$, and the relation symbol~$\leq$ is as
expected.
We initially only consider formulas that are conjunctions of positive
inequalities~$0 \leq \alpha_1 x_1 + \cdots + \alpha_k x_k+ \alpha_0$.
Equations~$s = t$ between real-valued terms can be translated to a
conjunction~$0 \leq s - t \wedge 0 \leq t - s$ of two inequalities,
while the handling of strict inequalities~$0 < s$, disjunctions, and
negated equations is discussed in \secref{sec:beyondIneqs}.


To interpolate linear inequalities, we rely on a classical result in
linear programming:
\begin{lemma}[Farkas' lemma~\cite{DBLP:books/daglib/0090562}]
  \label{lem:farkas}
  Let $A \in \R^{m \times n}$ be a real-valued matrix and $b \in \R^m$
  be a column vector. There is a column vector~$x \in \R^n$ such that
  $A x \leq b$ if and only if for every row vector~$y \in \R^m$ with
  $y \geq 0$ and $y A = 0$ it is the case that $yb \geq 0$.
\end{lemma}

The lemma implies a simple uniform representation of witnesses for the
unsatisfiability of linear inequalities over the reals: whenever a
conjunction~$\bigwedge_i 0 \leq t_i$ of linear inequalities is
unsatisfiable, we can form a non-negative linear
combination~$0 \leq \sum_i \alpha_i t_i$ of the inequalities in which
the right-hand side~$\sum_i \alpha_i t_i$ simplifies to a constant
negative value~$\beta < 0$. If the matrix~$A$ and the vector~$b$ in
the lemma only contain rational numbers, then also the vector~$y$ can
be chosen to be rational.

\begin{lemma}[Interpolants for linear
  inequalities~\cite{DBLP:journals/jsyml/Pudlak97,DBLP:journals/tcs/McMillan05}]
  \label{lem:LRA-interpolants}
  Suppose that
  \begin{equation*}
    \underbrace{\bigwedge_i 0 \leq t_i}_A ~\wedge~
    \underbrace{\bigwedge_j 0 \leq s_j}_B
  \end{equation*}
  is a conjunction of linear inequalities, and that the linear
  combination~$\sum_i \alpha_i t_i + \sum_j\beta_j s_j$, with
  $\alpha_i, \beta_j \geq 0$, is constant and negative. Then
  $I = \big(0 \leq \sum_i \alpha_i t_i\big)$, after simplification and
  removing terms with coefficient~$0$, is an LRA-separator for
  $A \wedge B$.
\end{lemma}

\begin{example}
  \label{ex:lra-1}
  We first illustrate this result using an example. Consider the
  conjunction
  \begin{equation*}
    \underbrace{0 \leq y - 1 \wedge 0 \leq z - x - 2y - 2}_A ~\wedge~
    \underbrace{0 \leq x \wedge 0 \leq -z + 2}_B~.
  \end{equation*}
  This sentence is unsatisfiable, since we can form a non-negative
  contradictory linear combination of the inequalities:
  \begin{equation*}
    0 \leq 2 \cdot (y - 1) + (z - x -2y - 2) + x + (-z + 2) ~\equiv~
    0 \leq -2~.
  \end{equation*}
  We can therefore construct an LRA-separator~$I$ by considering the
  corresponding linear combination of the inequalities in $A$. Note
  that the linear combination, before simplification, still contained
  the symbol~$y$ that is not common to $A$ and $B$. After
  simplification, this symbol has disappeared:
  \begin{equation*}
    0 \leq 2 \cdot (y - 1) + (z - x -2y - 2) ~\equiv~
    \underbrace{0 \leq z - x - 4}_I~.
  \end{equation*}
  \lipicsEnd
\end{example}

\begin{proof}[Proof of \lemref{lem:LRA-interpolants}]
  The properties~$A \models_{\text{LRA}} I$ and
  $I \wedge B \models_{\text{LRA}} \bot$ hold by construction. To see
  that
  $\sig(I) \subseteq (\sig(A) \cap \sig(B)) \cup \Sigma_{\text{LRA}}$,
  note that $I$ can only contain constant symbols that occur in
  $A$. Suppose now that some constant symbol~$x$ occurs in
  $\sum_i \alpha_i t_i$ with non-zero coefficient. Because
  $\sum_i \alpha_i t_i + \sum_j\beta_j s_j$ is constant, the
  symbol~$x$ then also has to occur in $\sum_j\beta_j s_j$, and $x$ is
  therefore in $\sig(A) \cap \sig(B)$.
\end{proof}

\begin{table}[p]
  \begin{tabularx}{\linewidth}{|X|}
    \hline
  \begin{subtable}{1.0\linewidth}
  \begin{equation*}
    \begin{array}{c@{\qquad}c}
      \prule
      {}
      {\naturalSeq{\Gamma}{0 \leq t}}
      {\ruleName{Hyp}}
      {$0 \leq t \in \Gamma$}
    &
       \prulee
       {\naturalSeq{\Gamma}{0 \leq s}}
       {\naturalSeq{\Gamma}{0 \leq t}}
       {\naturalSeq{\Gamma}{0 \leq \alpha s + \beta t}}
       {\ruleName{Comb}}
       {$\alpha, \beta > 0$}
    \\[3ex]
      \prule
      {\naturalSeq{\Gamma}{0 \leq \alpha}}
      {\naturalSeq{\Gamma}{\bot}}
      {\ruleName{Con}}
      {$\alpha < 0$}
    &
      \prule
      {\naturalSeq{\Gamma}{0 \leq s}}
      {\naturalSeq{\Gamma}{0 \leq t}}
      {\ruleName{Simp}}
      {$s \leadsto t$}
    \end{array}
  \end{equation*}

  \caption{Proof rules for linear real arithmetic. The
    arrow~$\leadsto$ denotes simplification of an expression.}
  \label{tab:lra-rules}
  \end{subtable}
  \\\hline
  \begin{subtable}{1.0\linewidth}
  \begin{equation*}
    \begin{array}{c@{\qquad}c}
      \prule
      {}
      {\naturalIntSeq{(A, B)}{0 \leq t}{t}}
      {\ruleNameInt{Hyp-A}}
      {$0 \leq t \in A$}
      &
      \prule
      {}
      {\naturalIntSeq{(A, B)}{0 \leq t}{0}}
      {\ruleNameInt{Hyp-B}}
      {$0 \leq t \in B$}
      \\[2.5ex]
      \multicolumn{2}{c}{
       \prulee
       {\naturalIntSeq{(A, B)}{0 \leq s}{s'}}
       {\naturalIntSeq{(A, B)}{0 \leq t}{t'}}
       {\naturalIntSeq{(A, B)}{0 \leq \alpha s + \beta t}{\alpha s' +
      \beta t'}}
       {\ruleNameInt{Comb}}
      {$\alpha, \beta > 0$}
      }
    \\[3.5ex]
      \prule
      {\naturalIntSeq{(A, B)}{0 \leq \alpha}{t'}}
      {\intSeq{(A, B)}{0 \leq t'}}
      {\ruleNameInt{Con}}
      {$\alpha < 0$}
      &
      \prulet
      {\naturalIntSeq{(A, B)}{0 \leq s}{s'}}
      {\naturalIntSeq{(A, B)}{0 \leq t}{t'}}
      {\ruleNameInt{Simp}}
      {
        (s \leadsto t, \\ \;s' \leadsto t')
      }
    \end{array}
  \end{equation*}

  \caption{Interpolating proof rules for linear real arithmetic.}
  \label{tab:lra-interpolating-rules}
  \end{subtable}
  \\\hline
  \begin{subtable}{1.0\linewidth}
  \begin{equation*}
    \begin{array}{c}
      \prulee
      {\intSeq{(A \cup \{ K_1\}, B)}{I}}
      {\intSeq{(A \cup \{ K_2\}, B)}{J}}
      {\intSeq{(A \cup \{ K_1 \vee K_2 \}, B)}{I \vee J}}
      {\ruleNameInt{Disj-A}}
      {}
      \\[3.5ex]
      \prulee
      {\intSeq{(A, B \cup \{ K_1\})}{I}}
      {\intSeq{(A, B \cup \{ K_2\})}{J}}
      {\intSeq{(A, B \cup \{ K_1 \vee K_2 \})}{I \wedge J}}
      {\ruleNameInt{Disj-B}}
      {}
    \end{array}
  \end{equation*}

  \caption{Interpolating proof rules for disjunctions.}
  \label{tab:disj-interpolating-rules}
  \end{subtable}
  \\\hline
  \begin{subtable}{1.0\linewidth}
  \begin{equation*}
    \begin{array}{l}
      \prulee
      {\naturalSeq{\Gamma \cup \{0 \leq n - x\}}{\bot}}
      {\naturalSeq{\Gamma \cup \{0 \leq x - n - 1\}}{\bot}}
      {\naturalSeq{\Gamma}{\bot}}
      {\ruleName{BnB}}
      {$x \in \sig(\Gamma), n \in \Z$}
      \\[3ex]
      \prulee
      {\naturalSeq{\Gamma \cup \{t = 0\}}{\bot}}
      {\naturalSeq{\Gamma \cup \{0 \leq t - 1\}}{\bot}}
      {\naturalSeq{\Gamma \cup \{0 \leq t\}}{\bot}}
      {\ruleName{Strengthen}}
      {}
      \\[3ex]
      \prule
      {\naturalSeq{\Gamma}{0 \leq \alpha_1 x_1 + \cdots + \alpha_k x_k
      + \alpha_0 }}
      {\naturalSeq{\Gamma}{0 \leq (\alpha_1/\beta) x_1 + \cdots + (\alpha_k/\beta) x_k
      + \lfloor\alpha_0 / \beta\rfloor}}
      {\ruleName{Div}}
      {$\beta = \gcd(\alpha_1, \ldots, \alpha_k), \beta > 0$}
    \end{array}
  \end{equation*}

  \caption{Proof rules for linear integer arithmetic.}
  \label{tab:lia-rules}
  \end{subtable}
  \\\hline
  \begin{subtable}{1.0\linewidth}
  \begin{equation*}
    \begin{array}{@{}l@{}}
      \prulee
      {\intSeq{(A \cup \{0 \leq n - x\}, B)}{I}}
      {\intSeq{(A \cup \{0 \leq x - n - 1\}, B)}{J}}
      {\intSeq{(A, B)}{I \vee J}}
      {\ruleNameInt{BnB-A}}
      {$x \in \sig(A), n \in \Z{}$}
      \\[4ex]
      \prulee
      {\intSeq{(A, B \cup \{0 \leq n - x\})}{I}}
      {\intSeq{(A, B \cup \{0 \leq x - n - 1\})}{J}}
      {\intSeq{(A, B)}{I \wedge  J}}
      {\ruleNameInt{BnB-B}}
      {$x \in \sig(B), n \in \Z$}
      \\[3.5ex]
      \prule
      {\naturalIntSeq{(A, B)}{0 \leq \alpha_1 x_1 + \cdots + \alpha_k x_k
      + \alpha_0 }{t'}}
      {\naturalIntSeq{(A, B)}{0 \leq (\alpha_1/\beta) x_1 + \cdots + (\alpha_k/\beta) x_k
      + \lfloor\alpha_0 / \beta\rfloor}{ t' \div  \beta }}
      {\ruleNameInt{Div}}
      {$\beta = \gcd(\alpha_1, \ldots, \alpha_k), \beta > 0$}
    \end{array}
  \end{equation*}

  \caption{Interpolating proof rules for linear integer arithmetic.}
  \label{tab:lia-interpolating-rules}
  \end{subtable}
  \\\hline
  \end{tabularx}

  \caption{Proof rules for linear real arithmetic (LRA) and linear
    integer arithmetic (LIA).}
\end{table}

\medskip
It is common to define this method of computing LRA-separators in
terms of proof rules. For this, we can first consider the proof rules
shown in \tabref{tab:lra-rules}, which form a calculus for deriving the
unsatisfiability of a set of
inequalities~\cite{DBLP:journals/tcs/McMillan05}. The proof rules are
defined in terms of \emph{sequents}~\shortNaturalSeq{\Gamma}{\phi}, in
which $\Gamma$ is a finite set of sentences (in our case, set of
inequalities), and $\phi$ is a single sentence. The semantics of a
sequent is that the conjunction of the formulas in $\Gamma$ entails
$\phi$. The proof rules express that the sequent underneath the bar
(the \emph{conclusion}) can be derived from the sequents  above the bar (the
\emph{premises}).

The rules~\ruleName{Hyp} and \ruleName{Comb} enable the derivation of
positive linear combinations of the given
inequalities~$\Gamma$. Rule~\ruleName{Con} detects unsatisfiable
inequalities, i.e., constant inequalities with a negative constant
term. The rule~\ruleName{Simp} simplifies inequalities with the help
of a rewriting relation~$\leadsto$, which can be implemented by
sorting the terms in an inequality and removing all terms with
coefficient~$0$~\cite{DBLP:conf/lpar/Rummer08}. An example proof is
given in \figref{fig:lra-example-proof}.

\begin{figure}[tb]
  \begin{gather*}
    \AxiomC{}
    \LeftLabel{\ruleName{Hyp}}
    \UnaryInfC{\naturalSeq{\Gamma}{0 \leq z - x - 2y - 2}}
    \AxiomC{}
    \LeftLabel{\ruleName{Hyp}}
    \UnaryInfC{\naturalSeq{\Gamma}{0 \leq x}}
    \LeftLabel{\ruleName{Comb}}
    \BinaryInfC{\naturalSeq{\Gamma}{0 \leq (z - x - 2y - 2) + x}}
    \LeftLabel{\ruleName{Simp}}
    \UnaryInfC{\naturalSeq{\Gamma}{0 \leq z -2y - 2}}
    \AxiomC{}
    \LeftLabel{\ruleName{Hyp}}
    \UnaryInfC{\naturalSeq{\Gamma}{0 \leq y - 1}}
    \LeftLabel{\ruleName{Comb}}
    \BinaryInfC{\naturalSeq{\Gamma}{0 \leq (z -2y - 2) + 2(y-1)}}
    \LeftLabel{\ruleName{Simp}}
    \UnaryInfC{\naturalSeq{\Gamma}{0 \leq z  - 4}}
    \UnaryInfC{$\mathcal{P}$}
    \DisplayProof
    \\[2ex]
    \hspace*{25ex}
    \AxiomC{$\mathcal{P}$}
    \AxiomC{}
    \LeftLabel{\ruleName{Hyp}}
    \UnaryInfC{\naturalSeq{\Gamma}{0 \leq -z + 2}}
    \LeftLabel{\ruleName{Comb}}
    \BinaryInfC{\naturalSeq{\Gamma}{0 \leq (z-4) + (-z+2)}}
    \LeftLabel{\ruleName{Simp}}
    \UnaryInfC{\naturalSeq{\Gamma}{0 \leq -2}}
    \LeftLabel{\ruleName{Con}}
    \UnaryInfC{\naturalSeq{\Gamma}{\bot}}
    \DisplayProof
  \end{gather*}
  
  \caption{Proof of unsatisfiability for
    $\Gamma = \{0 \leq y - 1, 0 \leq z - x - 2y - 2, 0 \leq x, 0 \leq
    -z + 2\}$ in \exref{ex:lra-1}.}
  \label{fig:lra-example-proof}
\end{figure}

It follows from \lemref{lem:farkas} that the calculus in
\tabref{tab:lra-rules} is sound and complete for showing
LRA-unsatisfiability: the sequent~$\shortNaturalSeq{\Gamma}{\bot}$ can
be derived if and only if $\Gamma$ is a finite LRA-unsatisfiable set
of inequalities.

\medskip
To derive interpolants, we augment the proof rules with
\emph{annotations} that reduce, in the root of a proof tree, to an
LRA-separator. The resulting interpolation proof rules are similar in
spirit to the interpolating proof rules for propositional logic
introduced in \refchapter{chapter:propositional}; like the
rules for propositional logic, interpolants are constructed
\emph{bottom-up,} starting at the leaves of a proof.  We follow the
style of interpolating proof rules in
\cite{DBLP:journals/tcs/McMillan05,DBLP:journals/corr/abs-1010-4422},
but streamline the formulation as we do not consider the combination
with other theories at the moment. The interpolating LRA proof rules
are given in \tabref{tab:lra-interpolating-rules}. For each of the proof
rules in \tabref{tab:lra-rules}, there are one or multiple interpolating
rules, now reasoning about two kinds of sequents. The simpler kind of
sequent has the shape~\shortIntSeq{(A, B)}{I} and contains two finite
sets~$A, B$ of sentences, as well as a single sentence~$I$. A
sequent~\shortIntSeq{(A, B)}{I} is considered valid if and only if $I$
is an LRA-separator for $\bigwedge A \wedge \bigwedge B$ according to
\defref{def:separator}. Validity of \shortIntSeq{(A, B)}{I} implies that
$A \cup B$ is LRA-unsatisfiable.

The second kind of sequent (also called \emph{inequality
  interpolation}~\cite{DBLP:journals/tcs/McMillan05}) has the shape
$\naturalIntSeq{(A, B)}{0 \leq t}{t'}$, again containing two
sets~$A, B$ of sentences, one inequality~$0 \leq t$, and one term~$t'$
forming the annotation. The intuition is that a sequent of this kind
corresponds to a non-interpolating
sequent~$\shortNaturalSeq{A \cup B}{0 \leq t}$, but in addition 
it records the contribution to $0 \leq t$ from $A$. A
sequent~$\naturalIntSeq{(A, B)}{0 \leq t}{t'}$ is valid if the
following conditions hold:
\begin{itemize}
\item $A \models_{\text{LRA}} 0 \leq t'$,
\item $B \models_{\text{LRA}} 0 \leq t - t'$, and
\item $\sig(t') \subseteq \sig(A) \cup \Sigma_{\text{LRA}}$ and
  $t - t' \leadsto s$ with
  $\sig(s) \subseteq \sig(B) \cup \Sigma_{\text{LRA}}$.
\end{itemize}
The last condition expresses that $t'$ is a term only containing
symbols from $A$, while the difference $t - t'$ can be simplified to a
term only containing symbols from $B$ (ignoring symbols
in~$\Sigma_{\text{LRA}}$). In other words, the annotation splits the
inequality~$0 \leq t$ into an $A$-part and a $B$-part. The proof rules
in \tabref{tab:lra-interpolating-rules} are formulated in such a way
that every sequent that can be derived is valid; for this,
\ruleNameInt{Hyp-A} and \ruleNameInt{Hyp-B} record whether an
inequality comes from $A$ or from $B$, \ruleNameInt{Comb} propagates
this information, and \ruleNameInt{Simp} simplifies the
annotation~$t'$ along with the derived formula~$0 \leq t$. The rule
\ruleNameInt{Con} concludes the derivation of an LRA-separator: if an
unsatisfiable inequality~$0 \leq t$ has been derived, the
annotation~$t'$ gives rise to a separator~$0 \leq t'$. The
interpolating version of the proof in \figref{fig:lra-example-proof}
is given in \figref{fig:lra-interpolating-example-proof} and produces
the same separator~$0 \leq z - x - 4$ as the reasoning in
\exref{ex:lra-1}.

\begin{figure}[tb]
  \begin{gather*}
    \AxiomC{}
    \LeftLabel{\ruleNameInt{Hyp-A}}
    \UnaryInfC{\naturalIntSeq{\Gamma^I}{0 \leq z - x - 2y - 2}{z - x - 2y - 2}}
    \AxiomC{}
    \LeftLabel{\ruleNameInt{Hyp-B}}
    \UnaryInfC{\naturalIntSeq{\Gamma^I}{0 \leq x}{0}}
    \LeftLabel{\ruleNameInt{Comb}}
    \BinaryInfC{\naturalIntSeq{\Gamma^I}{0 \leq (z - x - 2y - 2) + x}{(z - x - 2y - 2) + 0}}
    \LeftLabel{\ruleNameInt{Simp}}
    \UnaryInfC{\naturalIntSeq{\Gamma^I}{0 \leq z -2y - 2}{z - x - 2y - 2}}
    \UnaryInfC{$\mathcal{P}_1$}
    \DisplayProof
    \\[2ex]
    \hspace*{5ex}
    \AxiomC{$\mathcal{P}_1$}
    \AxiomC{}
    \LeftLabel{\ruleNameInt{Hyp-A}}
    \UnaryInfC{\naturalIntSeq{\Gamma^I}{0 \leq y - 1}{y - 1}}
    \LeftLabel{\ruleNameInt{Comb}}
    \BinaryInfC{\naturalIntSeq{\Gamma^I}{0 \leq (z -2y - 2) + 2(y-1)}{(z - x - 2y - 2) + 2(y - 1)}}
    \LeftLabel{\ruleNameInt{Simp}}
    \UnaryInfC{\naturalIntSeq{\Gamma^I}{0 \leq z  - 4}{z - x - 4}}
    \UnaryInfC{$\mathcal{P}_2$}
    \DisplayProof
    \\[2ex]
    \hspace*{15ex}
    \AxiomC{$\mathcal{P}_2$}
    \AxiomC{}
    \LeftLabel{\ruleNameInt{Hyp-B}}
    \UnaryInfC{\naturalIntSeq{\Gamma^I}{0 \leq -z + 2}{0}}
    \LeftLabel{\ruleNameInt{Comb}}
    \BinaryInfC{\naturalIntSeq{\Gamma^I}{0 \leq (z-4) + (-z+2)}{(z - x - 4) + 0}}
    \LeftLabel{\ruleNameInt{Simp}}
    \UnaryInfC{\naturalIntSeq{\Gamma^I}{0 \leq -2}{z - x - 4}}
    \LeftLabel{\ruleNameInt{Con}}
    \UnaryInfC{\intSeq{\Gamma^I}{0 \leq z - x - 4}}
    \DisplayProof
  \end{gather*}
  
  \caption{Interpolating version of the proof in
    \figref{fig:lra-example-proof}.  We use the
    abbreviation~$\Gamma^I = (A, B) = (\{0 \leq y - 1, 0 \leq z - x -
    2y - 2\}, \{0 \leq x, 0 \leq -z + 2\})$.}
  \label{fig:lra-interpolating-example-proof}
\end{figure}

\subsubsection{Beyond Conjunctions of
  Non-strict Inequalities}
\label{sec:beyondIneqs}

Our discussion so far was restricted to conjunctions of non-strict
inequalities~$0 \leq t$ and equations~$s = t$. The formulation in
terms of interpolating proof rules leads to a flexible framework,
however, that can be extended to include further operators:
\begin{itemize}
\item \emph{Boolean structure} in formulas can be handled by
  introducing further proof rules according to the principles
  developed in \refchapter{chapter:prooftheory}. For completeness, we
  show the interpolating proof rules for splitting disjunctions in the
  $A$- and $B$-sets in \tabref{tab:disj-interpolating-rules}. The
  rules introduce a disjunction or conjunction in the interpolant,
  depending on whether a formula from $A$ or from $B$ is
  considered. Similar rules can be defined for conjunctions and
  negations.

  It is also generally possible to combine an interpolation procedure
  for a theory~$T$ with the propositional methods from
  \refchapter{chapter:propositional}, resulting in interpolation
  procedures that can handle arbitrary quantifier-free
  formulas~\cite{DBLP:journals/tcs/McMillan05}.\todo{more details}
  This is the approach implemented in interpolating SMT
  solvers, which combine a SAT solver for efficient propositional
  reasoning with theory solvers that only see conjunctions of theory
  literals~\cite{BSST21}.
\item 
\emph{Strict inequalities} $0 < t$ in formulas can be handled by
replacing Farkas' lemma with Motzkin's transposition
theorem~\cite{DBLP:books/daglib/0090562}, which gives a complete
criterion for the unsatisfiability of a conjunction of strict and
non-strict inequalities. Based on the transposition theorem, it is
possible to formulate a lemma similar to \lemref{lem:LRA-interpolants}
for interpolation in the presence of both strict and non-strict
inequalities~\cite{DBLP:journals/jsc/RybalchenkoS10}.
An alternative path is to add separate proof rules for equations and
uninterpreted functions~\cite{DBLP:journals/tcs/McMillan05}, which can
then be used to reason about strict inequalities~$0 < t$ by rewriting
them to a conjunction~$0 \leq t \wedge 0 \not= t$.
\item \emph{Negated equations} $s \not= t$ can be replaced by
  disjunctions $0 < s - t \vee 0 < t - s$, which can then be handled
  by the rules for disjunctions and strict inequalities.
\end{itemize}

\subsubsection{Linear Integer Arithmetic (LIA)}

The computation of Craig interpolants becomes a bit more involved when
moving from real numbers to the integers. The good news is that the
proof rules presented in the previous section remain sound also for
integer arithmetic; as a result, also any LRA-separator derived using
the rules in \tabref{tab:lra-interpolating-rules} is a correct
LIA-separator~\cite{DBLP:journals/tcs/McMillan05}. Initial software
verification tools based on Craig interpolation therefore resorted to
pure LRA reasoning, even for programs operating on integers or bounded
machine integers.
The LRA proof rules are not complete for integer reasoning, of course,
since there are formulas that are satisfiable over the reals but
unsatisfiable over integers. To prove or interpolate such formulas,
SMT solvers augment the LRA rules with additional LIA-specific rules,
which are, however, harder to turn into interpolating rules.

For \emph{linear integer arithmetic,} LIA, we consider formulas that
can contain linear
inequalities~$0 \leq \alpha_1 x_1 + \cdots + \alpha_k x_k + \alpha_0$,
where $\alpha_0, \alpha_1, \ldots, \alpha_k \in \Z$ are now integers
and $x_1, \ldots, x_k$ are constant symbols of sort~$\mathit{Int}$, as
well as divisibility statements~$k \mid \cdot$ for every $k \in \N$
(see \exref{ex:lia-general-qf-interpolation}). The semantics of such
formulas is as expected.

\begin{example}
  \label{ex:exponential-LIA-interpolant}
  Consider the following families of formulas (where
  $n \in \N, n > 1$)~\cite{DBLP:conf/lpar/KroeningLR10}:
  \begin{equation*}
    A_n ~=~ 0 \leq y + 2nx + n - 1 \wedge 0 \leq -y-2nx,
    \qquad
    B_n ~=~  0 \leq y + 2nz - 1 \wedge 0 \leq   -y - 2nz + n ~.
  \end{equation*}
  Using more intuitive notation, the formulas can be written as
  $A_n \equiv -n < y+2nx \leq 0$ and $B_n \equiv 0 < y + 2nz \leq n$,
  and it is easy to see that $A_n \wedge B_n$ is satisfiable over the
  reals, but unsatisfiable over the integers (for any $n > 1$). Up to
  equivalence, the only LIA-separator for $A_n \wedge B_n$ is the
  sentence:
  \begin{equation}
    \label{eq:LIA-separator}
    I_n ~=~
    (2n \mid y)
    \vee
    (2n \mid y + 1)
    \vee
    (2n \mid y + 2)
    \vee \cdots \vee
    ( 2n \mid  y + n - 1 ) ~.
  \end{equation}
  Given a LIA signature providing only inequalities and divisibility
  statements, the size of separators is linear in $n$ and therefore
  exponential in the size of $A_n, B_n$.
  \lipicsEnd
\end{example}

There are three different kinds of LIA proof rules commonly
implemented in solvers, in addition to the LRA rules that we have
already seen in the previous section:
\begin{itemize}
\item \emph{Branch-and-bound (BnB)}~\cite{DBLP:books/daglib/0090562}: a
  solver first checks the satisfiability of a formula~$\phi$ over the
  reals instead of the integers, an approach known as
  \emph{relaxation.} If $\phi$ is found to be LRA-unsatisfiable, it is
  necessarily also LIA-unsatisfiable; if, on the other hand, a
  solution over the real numbers is found, the solver checks whether
  there are any symbols~$x$ that were assigned a non-integer
  value~$\alpha \in \R\setminus \Z$. When such an~$x$ exists,
  the search branches and continues separately for the two
  cases~$x \leq \lfloor \alpha \rfloor$ and
  $x \geq \lceil \alpha \rceil$.
  The proof rule~\ruleName{BnB} representing the branch-and-bound
  strategy is shown in \tabref{tab:lia-rules}.

  While branch-and-bound is an important proof rule used in solvers,
  and tends to perform well on many formulas, it does not give rise to
  a complete calculus for LIA. This can be seen when considering the
  formulas in \exref{ex:exponential-LIA-interpolant}, which have
  LRA-solutions of unbounded size, but no LIA-solutions.  Regardless
  of how often the \ruleName{BnB} rule is applied, there are still
  proof branches left in which the formulas in $\Gamma$ are
  LRA-satisfiable but LIA-unsatisfiable.
  
\item \emph{Strengthening}~\cite{omega,DBLP:conf/lpar/Rummer08}:
  similarly as with branch-and-bound, inequalities~$0 \leq t$ to be
  solved can be used to split the search into the cases~$t = 0$ and
  $1 \leq t$, expressing the fact that an integer term~$t$ cannot have
  a value strictly between~$0$ and $1$. Strengthening is the core rule
  used by the Omega test~\cite{omega}, in which the boundary
  case~$t = 0$ is also called a ``splinter''. By suitably restricting
  the number of splinters that are introduced, the Omega test uses
  strengthening for quantifier elimination in LIA.

  The corresponding rule~\ruleName{Strengthen} is shown in
  \tabref{tab:lia-rules}. Strengthening can be used to obtain a complete
  calculus for LIA, but it requires a separate set of rules to solve
  the equations that stem from the
  splinters~\cite{omega,DBLP:conf/lpar/Rummer08}; such rules for
  equations are not shown here. In
  \exref{ex:exponential-LIA-interpolant}, the \ruleName{Strengthen}
  rule can be applied $n$ times to the first inequality in $A_n$ to
  obtain an LRA-unsatisfiable set of inequalities, which can be
  proven using the rules in \tabref{tab:lra-rules}, as well as $n$
  splinters that can be handled using equational reasoning.

\begin{figure}[tb]
  \begin{gather*}
    \AxiomC{}
    \LeftLabel{\ruleName{Hyp}}
    \UnaryInfC{\naturalSeq{\Gamma}{0 \leq y + 2nx + n - 1}}
    \AxiomC{}
    \LeftLabel{\ruleName{Hyp}}
    \UnaryInfC{\naturalSeq{\Gamma}{0 \leq   -y - 2nz + n}}
    \LeftLabel{\ruleName{Comb}}
    \BinaryInfC{\naturalSeq{\Gamma}{0 \leq  (y + 2nx + n - 1) + (-y - 2nz + n)}}
    \LeftLabel{\ruleName{Simp}}
    \UnaryInfC{\naturalSeq{\Gamma}{0 \leq   2nx - 2nz + 2n - 1}}
    \LeftLabel{\ruleName{Div}}
    \UnaryInfC{\naturalSeq{\Gamma}{0 \leq   x - z + 0}}
    \LeftLabel{\ruleName{Simp}}
    \UnaryInfC{\naturalSeq{\Gamma}{0 \leq   x - z}}
    \UnaryInfC{$\mathcal{P}_1$}
    \DisplayProof
    \\[2ex]
    \hspace*{15ex}
    \AxiomC{}
    \LeftLabel{\ruleName{Hyp}}
    \UnaryInfC{\naturalSeq{\Gamma}{0 \leq -y-2nx}}
    \AxiomC{}
    \LeftLabel{\ruleName{Hyp}}
    \UnaryInfC{\naturalSeq{\Gamma}{0 \leq y + 2nz - 1}}
    \LeftLabel{\ruleName{Comb}}
    \BinaryInfC{\naturalSeq{\Gamma}{0 \leq  (-y-2nx) + (y + 2nz - 1)}}
    \LeftLabel{\ruleName{Simp}}
    \UnaryInfC{\naturalSeq{\Gamma}{0 \leq   2nz - 2nx - 1}}
    \LeftLabel{\ruleName{Div}}
    \UnaryInfC{\naturalSeq{\Gamma}{0 \leq   z - x - 1}}
    \UnaryInfC{$\mathcal{P}_2$}
    \DisplayProof
    \\[2ex]
    \hspace*{20ex}
    \AxiomC{$\mathcal{P}_1$}
    \AxiomC{$\mathcal{P}_2$}
    \LeftLabel{\ruleName{Comb}}
    \BinaryInfC{\naturalSeq{\Gamma}{0 \leq  (x-z) + (z-x - 1)}}
    \LeftLabel{\ruleName{Simp}}
    \UnaryInfC{\naturalSeq{\Gamma}{0 \leq  -1}}
    \LeftLabel{\ruleName{Con}}
    \UnaryInfC{\naturalSeq{\Gamma}{\bot}}
    \DisplayProof
  \end{gather*}
  
  \caption{Schematic cutting planes proof of unsatisfiability for the
    formulas
    $\Gamma = \{0 \leq y + 2nx + n - 1, 0 \leq -y-2nx, 0 \leq y + 2nz
    - 1, 0 \leq -y - 2nz + n\}$ in
    \exref{ex:exponential-LIA-interpolant}.}
  \label{fig:cp-proof}
\end{figure}

\item
  \emph{Cuts}~\cite{DBLP:books/daglib/0090562,DBLP:conf/cav/DutertreM06}
  allow solvers to round the constant term in inequalities in such a
  way that no integer solutions are lost. Cuts can be seen as a
  targeted application of the strengthening rule, restricting
  application to cases in which the introduced equation~$t = 0$ is
  immediately unsatisfiable. There are different strategies available
  for cuts, including Gomory
  cuts~\cite{DBLP:books/daglib/0090562,DBLP:conf/cav/DutertreM06} and
  the cuts-from-proofs method~\cite{DBLP:journals/fmsd/DilligDA11},
  both often used in solvers.

  A general rule~\ruleName{Div} implementing cuts is shown as the
  third rule in \tabref{tab:lia-rules}. The rule~\ruleName{Div} is
  applicable if the coefficients of the non-constant terms in an
  inequality have some factor~$\beta>0$ in common. In this case, all
  coefficients in the inequality can be divided by $\beta$ while
  rounding down the constant term~$\alpha_0 / \beta$, strengthening
  the inequality.  Cuts give rise to a complete calculus for LIA and
  have the advantage that no branching is needed; proofs including
  cuts are also called \emph{cutting planes proofs.} To obtain even
  better performance in practice, many solvers combine cuts with the
  branch-and-bound method.

  We show a cutting planes proof for the formulas in
  \exref{ex:exponential-LIA-interpolant} in \figref{fig:cp-proof}. The
  proof works for any positive value~$n$ and requires two applications
  of the \ruleName{Div} rule, after which a contradiction can be
  derived using \ruleName{Comb}.
\end{itemize}

All three kinds of proof rules can be made interpolating. We show the
interpolating versions of \emph{branch-and-bound} and \emph{cuts} in
\tabref{tab:lia-interpolating-rules}, using the same sequent notation as
for the interpolating LRA rules in
\tabref{tab:lra-interpolating-rules}. Interpolating versions of the
\emph{strengthening} rule are somewhat more involved, due to the need
to introduce also interpolating rules for linear
equations~\cite{DBLP:journals/fmsd/JainCG09}, but can be derived as
well~\cite{iprincess2011,DBLP:journals/corr/abs-1010-4422}.

The simplest rule to interpolate is branch-and-bound, since the
branching performed by \ruleName{BnB} can be seen as the splitting of
a disjunction~$0 \leq n - x \vee 0 \leq x - n - 1$ (which is a valid
formula in LIA). The two interpolating rules corresponding
to~\ruleName{BnB} are \ruleNameInt{BnB-A} and \ruleNameInt{BnB-B}:
depending on whether the constant~$x$ is a symbol from $A$ or from
$B$, one or the other rule has to be used. The
rules~\ruleNameInt{BnB-A} and \ruleNameInt{BnB-B} correspond to the
rules~\ruleNameInt{Disj-A} and \ruleNameInt{Disj-B} for splitting
disjunctions; like the rules for disjunctions, they derive the overall
interpolant for $(A, B)$ by combining the sub-interpolants~$I, J$
using either a disjunction or a conjunction. If $x$ occurs in both $A$
and $B$, either rule can be used, but the rules will, in general, lead
to different interpolants.

Interpolating cutting planes proofs is more involved, since the
inequalities tightened in a proof using the rule~\ruleName{Div} might
contain some symbols that only occur in $A$, and some symbols that
only occur in $B$, and therefore cannot be associated uniquely with
either side. It is possible, however, to extend the annotation
approach for LRA interpolation to LIA by allowing the annotation~$t'$
in a sequent~$\naturalIntSeq{(A, B)}{0 \leq t}{t'}$ to contain the
operator~$\div$ for integer division with remainder, provided that the
denominator is an integer
constant~\cite{DBLP:journals/corr/abs-1010-4422,DBLP:journals/jsyml/Pudlak97}. We
use the Euclidean definition of division with
remainder~\cite{DBLP:books/daglib/0070572}, in which $a \div b$ for
$a, b \in \Z$, $b \not= 0$ is defined to be the unique
number~$q \in \Z$ such that $a = b \cdot q + r$ with $0 \leq r < |b|$.
To handle integer division, it is also necessary to generalize the
\ruleNameInt{Simp} rule by adding a case for pulling terms out of the
scope of $\div$, namely
$(\alpha s + t) \div \beta \leadsto (\alpha / \beta)s + (t \div
\beta)$ whenever $\beta \mid \alpha$. This rewrite is correct, since
the equation~$(\beta x + y) \div \beta = x + (y \div \beta)$ holds for
Euclidian division, provided that $\beta \not= 0$.

\begin{figure}[p]
  \centering
  \rotatebox{90}{
    \begin{minipage}{1.15\linewidth}
  \begin{gather*}
    \AxiomC{}
    \LeftLabel{\ruleNameInt{Hyp-A}}
    \UnaryInfC{\naturalIntSeq{\Gamma^I}{0 \leq y + 2nx + n - 1}{y + 2nx + n - 1}}
    \AxiomC{}
    \LeftLabel{\ruleNameInt{Hyp-B}}
    \UnaryInfC{\naturalIntSeq{\Gamma^I}{0 \leq   -y - 2nz + n}{0}}
    \LeftLabel{\ruleNameInt{Comb}}
    \BinaryInfC{\naturalIntSeq{\Gamma^I}{0 \leq  (y + 2nx + n - 1) + (-y - 2nz + n)}{(y + 2nx + n - 1) + 0}}
    \LeftLabel{\ruleNameInt{Simp}}
    \UnaryInfC{\naturalIntSeq{\Gamma^I}{0 \leq   2nx - 2nz + 2n - 1}{y + 2nx + n - 1}}
    \LeftLabel{\ruleNameInt{Div}}
    \UnaryInfC{\naturalIntSeq{\Gamma^I}{0 \leq   x - z + 0}{(y + 2nx + n - 1) \div 2n}}
    \LeftLabel{\ruleNameInt{Simp}}
    \UnaryInfC{\naturalIntSeq{\Gamma^I}{0 \leq   x - z}{x + (y + n - 1) \div 2n}}
    \UnaryInfC{$\mathcal{P}_1$}
    \DisplayProof
    \\[2ex]
    \hspace*{10ex}
    \AxiomC{}
    \LeftLabel{\ruleNameInt{Hyp-A}}
    \UnaryInfC{\naturalIntSeq{\Gamma^I}{0 \leq -y-2nx}{-y-2nx}}
    \AxiomC{}
    \LeftLabel{\ruleNameInt{Hyp-B}}
    \UnaryInfC{\naturalIntSeq{\Gamma^I}{0 \leq y + 2nz - 1}{0}}
    \LeftLabel{\ruleNameInt{Comb}}
    \BinaryInfC{\naturalIntSeq{\Gamma^I}{0 \leq  (-y-2nx) + (y + 2nz - 1)}{(-y-2nx) + 0}}
    \LeftLabel{\ruleNameInt{Simp}}
    \UnaryInfC{\naturalIntSeq{\Gamma^I}{0 \leq   2nz - 2nx - 1}{-y-2nx}}
    \LeftLabel{\ruleNameInt{Div}}
    \UnaryInfC{\naturalIntSeq{\Gamma^I}{0 \leq   z - x - 1}{(-y-2nx) \div 2n}}
    \LeftLabel{\ruleNameInt{Simp}}
    \UnaryInfC{\naturalIntSeq{\Gamma^I}{0 \leq   z - x - 1}{-x + (-y) \div 2n}}
    \UnaryInfC{$\mathcal{P}_2$}
    \DisplayProof
    \\[2ex]
    \hspace*{7ex}
    \AxiomC{$\mathcal{P}_1$}
    \AxiomC{$\mathcal{P}_2$}
    \LeftLabel{\ruleNameInt{Comb}}
    \BinaryInfC{\naturalIntSeq{\Gamma^I}{0 \leq  (x-z) + (z-x - 1)}{(x + (y + n - 1) \div 2n) + (-x + (-y) \div 2n)}}
    \LeftLabel{\ruleNameInt{Simp}}
    \UnaryInfC{\naturalIntSeq{\Gamma^I}{0 \leq  -1}{(y + n - 1) \div 2n + (-y) \div 2n}}
    \LeftLabel{\ruleNameInt{Con}}
    \UnaryInfC{\intSeq{\Gamma^I}{0 \leq (y + n - 1) \div 2n + (-y) \div 2n}}
    \DisplayProof
  \end{gather*}
    \end{minipage}
  }
  
  
  \caption{Interpolating proof for the formulas in
    \exref{ex:exponential-LIA-interpolant}:
    $\Gamma^I = (A, B) = (\{0 \leq y + 2nx + n - 1,\linebreak 0 \leq -y-2nx\},
    \{0 \leq y + 2nz - 1, 0 \leq -y - 2nz + n\})$.}
  \label{fig:cp-interpolating-proof}
\end{figure}

The rule~\ruleNameInt{Div} in \tabref{tab:lia-interpolating-rules} is
deceivingly simple and uses $\div$ to divide the annotation~$t'$ in a
sequent by the same number~$\beta$ as the main inequality.  The effect
of \ruleNameInt{Div} is best illustrated using an
example. \figref{fig:cp-interpolating-proof} shows the interpolating
version of the proof in \figref{fig:cp-proof}, proving that the
formulas in \exref{ex:exponential-LIA-interpolant} are unsatisfiable
for any positive $n$. The two applications of \ruleNameInt{Div}
introduce a division~$\mbox{} \div 2n$ in the annotation term. The
annotations are subsequently simplified by applying the
rule~\ruleNameInt{Simp}, which reduces the terms~$2nx$ and $-2nx$ to
$x$ and $-x$, respectively, removing the symbol~$x$ from the scope of
$\div$. Since $x$ and $-x$ cancel each other out, the final
LIA-separator is the
formula~$0 \leq (y + n - 1) \div 2n + (-y) \div 2n$ that only contains
the symbol~$y$ that is common to $A$ and $B$. The reader can verify
that this separator is equivalent to $I_n$ in
\exref{ex:exponential-LIA-interpolant} but exponentially more
succinct thanks to the use of $\div$.

The interesting question, of course, is why the rule~\ruleNameInt{Div}
always produces correct interpolants. Adapted to the case of LIA, the
conditions to be satisfied by an interpolating
sequent~$\naturalIntSeq{(A, B)}{0 \leq t}{t'}$ are:
\begin{itemize}
\item $A \models_{\text{LIA}} 0 \leq t'$,
\item $B \models_{\text{LIA}} 0 \leq t - t'$, and
\item $\sig(t') \subseteq \sig(A) \cup \Sigma_{\text{LIA}}$ and
  $t - t' \leadsto s$ with
  $\sig(s) \subseteq \sig(B) \cup \Sigma_{\text{LIA}}$.
\end{itemize}
We can check, using purely arithmetic reasoning, that the first two
conditions hold in any sequent~$\naturalIntSeq{(A, B)}{0 \leq t}{t'}$
derived using the rules in \tabref{tab:lra-interpolating-rules} and
rule~\ruleNameInt{Div}. To see that the last conditions are preserved
by \ruleNameInt{Div}, note that the ability to rewrite~$t - t'$ to a
term~$s$ with $\sig(s) \subseteq \sig(B) \cup \Sigma_{\text{LIA}}$
implies that whenever a term~$\alpha x$ occurs in $t$, for some $x$
that is not present in $B$, there has to be a corresponding
term~$\alpha x$ in $t'$ canceling out $\alpha x$. When
applying~\ruleNameInt{Div}, it then has to be the case that $\alpha$
is a multiple of the denominator~$\beta$. An application of the
rule~\ruleNameInt{Simp} can therefore simplify the
annotation~$(\cdots +\alpha x + \cdots) \div \beta$ to
$(\alpha / \beta)x + (\cdots) \div \beta$.

More generally, after applying \ruleNameInt{Simp}, it will never be
the case that symbols that do not occur in $B$ are in the scope of the
$\div$ operator in the annotation~$t'$. Despite the use of $\div$ in
the annotations, the interpolating proof rules are able to eliminate
symbols that are present in $A$ but not in $B$ from annotations in
more or less the same way as for LRA, and the derived interpolants are
correct.


\subsection{Further Interpolation Approaches}

\subsubsection{Proof-Based Interpolation}
\label{sec:proofBased}

The extraction of (quantifier-free) $T$-separators from proofs is
possible not only for LRA and LIA, but also for many other
theories. 
The interpolating proof rules for LRA can be extended to handle also
the theory of Equality and Uninterpreted Functions
(EUF)~\cite{DBLP:journals/tcs/McMillan05}.
%
The LRA calculus can also be extended to non-linear arithmetic
(Non-linear Real Arithmetic, NRA), replacing the witnesses produced by
Farkas' lemma by Sum-Of-Squares (SOS)
witnesses~\cite{DBLP:conf/cav/DaiXZ13,DBLP:conf/fm/WuWXLZG24}.  NRA
interpolants can also be computed with the help of the
model-constructing satisfiability calculus
MCSAT~\cite{10.1007/978-3-030-81688-9_13}.

Several interpolating calculi have been presented for the theory of
arrays~\cite{DBLP:conf/rta/BruttomessoGR11,DBLP:conf/popl/TotlaW13,DBLP:conf/cade/HoenickeS18}. As
illustrated in \exref{ex:arrays}, one difficulty to be solved with
arrays is the fact that the standard version of arrays does not admit
(plain or general) quantifier-free interpolation. This problem can be
addressed by extending the theory by a function symbol~$\mathit{diff}$
satisfying the following axiom:
\begin{equation*}
  \forall M, M'.\; \big( M \not= M' \to
  \mathit{select}(M, \mathit{diff}(M, M')) \not=
  \mathit{select}(M', \mathit{diff}(M, M')) \big)~.
\end{equation*}
The axiom implies extensionality of arrays, i.e., two arrays are equal
if and only if they have the same contents.  In \exref{ex:arrays}, we
can then construct the quantifier-free reverse interpolant
\begin{equation*}
  M' = \mathit{store}\big(M, \mathit{diff}(M, M'),
  \mathit{select}(M', \mathit{diff}(M, M'))\big)~.
\end{equation*}

\subsubsection{Graph-Based Interpolation}
\label{sec:euf}

\contents{
\item Equality and Uninterpreted Functions (EUF)
\item Congruence closure and E-graphs
\item Recursive extraction of interpolants
\item From one graph, many interpolants can be extracted
\item (connect to chapter about automated reasoning)
\item Other examples: difference-bound constraints
}

In the context of SMT, the special case of the satisfiability of a
formula in the empty theory~$T_\emptyset$ is commonly called
\emph{satisfiability modulo Equality and Uninterpreted Functions
  (EUF).} EUF is in SMT solvers usually handled with the help of the
congruence closure algorithm~\cite{DBLP:journals/jar/BachmairTV03},
which decides whether a given equation~$s = t$ follows from a finite
set~$E$ of ground equations, i.e., $E \models_{T_\emptyset} s =
t$. For this, congruence closure constructs an equality graph (the
\emph{e-graph}) in which the terms and sub-terms of $E$ are the nodes,
the equations~$E$ and inferred equalities are edges, and a second set
of edges is used to represent the sub-term relation. Congruence
closure can show the unsatisfiability of a formula~$\phi$ by inferring
that some equation~$s = t$ follows from the equations present in
$\phi$, while $\phi$ also contains the negated equation~$s \not= t$. The
e-graph represents the proof explaining the unsatisfiability of a formula.

In case of an unsatisfiable conjunction~$A \wedge B$ of equations,
separators can be extracted from the
e-graph~\cite{DBLP:journals/corr/abs-1111-5652}. This is done by
summarizing paths in the graph; if a path is created only by equations
from the formula~$A$, it can be represented by a single equation that
is entailed by $A$, and similarly for equations from $B$. Reasoning in
terms of congruences ($s = t$ implies $f(s) = f(t)$) gives rise to
implications in the computed separators.


\subsubsection{Interpolation by Reduction}

\contents{
\item Algebraic Data-types
\item Translation of an interpolation problem to some domain in which
  interpolation procedures are available
\item Other examples: bit-vectors (bit-blasting, int-blasting),
  arrays, linked lists
}

A further  paradigm for computing Craig interpolants is
\emph{reduction}~\cite{DBLP:conf/sigsoft/KapurMZ06}. Reduction is
applicable whenever direct interpolation of a conjunction~$A \wedge B$
in some theory~$T$ is difficult, but it is possible to translate
$A \wedge B$ to a conjunction~$A' \wedge B'$ in a 
theory~$T'$ in which interpolants can be computed. The
$T'$-separator~$I'$ then has to be translated back to obtain a
$T$-separator~$I$ for the original conjunction~$A \wedge B$; in the
special case that $T$ is an extension of $T'$, back-translation is
unnecessary, since $I'$ is already a correct $T$-separator for
$A \wedge B$.

The reduction approach can be applied to a wide range of theories. The
theory of arrays and the theory of sets with finite cardinality
constraints can be reduced to EUF, while the theory of multisets can
be reduced to the combined theory of EUF and
LIA~\cite{DBLP:conf/sigsoft/KapurMZ06}. Since reduction can introduce
quantifiers, however, it is in general necessary to post-process
interpolants using a quantifier elimination procedure for
quantifier-free interpolation. In follow-up work, it was shown that
the theory of algebraic data types with size constraints can be
reduced to the combined theory of EUF and
LIA~\cite{DBLP:conf/synasc/HojjatR17}, and that the theory of
bit-vectors can be reduced to
LIA~\cite{DBLP:journals/fmsd/BackemanRZ21}. The latter reduction
happens lazily with the help of tailor-made proof rules that
successively translate bit-vector formulas to LIA formulas.

\section{Verification Algorithms Based on Interpolation}
\label{sec:verification}

We now introduce some of the main verification algorithms based on Craig
interpolation, starting with the SAT-based model checking algorithm
proposed by McMillan in 2003~\cite{DBLP:conf/cav/McMillan03}. We
present the algorithms in our setting of quantifier-free interpolation
in a theory~$T$, although the earlier papers focus on
propositional interpolation.

\subsection{BMC: Bounded Model Checking}

Interpolation-based verification algorithms were inspired by
\emph{Bounded Model Checking~(BMC),} in which the behavior of a system
during the first $k$ execution steps is analyzed with the help of SAT
or SMT solvers~\cite{DBLP:conf/tacas/BiereCCZ99}. The insight that led
to interpolation-based algorithms is that interpolants can be
extracted for the formulas considered in BMC, this way turning
\emph{bounded} into \emph{unbounded} model checking.

To introduce BMC, we fix a theory~$T$ and assume that states of the
system to be verified are described by a
vector~$\bar x = (x_1, \ldots, x_n)$ of constants; in the original
SAT-based formulation, $\bar x$ is a vector of propositional
variables. We use the notation~$\bar x' = (x'_1, \ldots, x'_n)$ for a
fresh primed copy of the constants, and similarly~$\bar x_i$ for
copies indexed by $i$.
BMC operates on \emph{transition systems} represented using formulas
that model, for instance, sequential hardware circuits.  A transition
system~$(B, \tr)$ is given by two sentences~$B[\bar x]$ and
$\tr[\bar x, \bar x']$ defining the \emph{initial states} and the
\emph{transition relation,} respectively. The notation~$B[\bar x]$
expresses that the constants~$\bar x$ can occur in $B$, and
$B[\bar y]$ that the constants~$\bar x$ are replaced by another
vector~$\bar y$ of terms. Intuitively, any valuation of $\bar x$
satisfying $B[\bar x]$ is considered an initial state of the system,
and valuations of $\bar x, \bar x' $ that satisfy
$\tr[\bar x, \bar x']$ represent a transition from $\bar x$ to
$\bar x'$.

A further sentence~$E[\bar x]$ is introduced to represent \emph{error
  states.}  The goal of verification is to decide whether any error
states are reachable from the initial states.  The verification of
partial correctness properties, such as the ones discussed in
\secref{sec:motivating}, can be reduced to this form, modeling the
states in which postconditions or assertions are violated using the
formula~$E$.
In BMC, the reachability of error states is
characterized by unfolding the transition relation a bounded number of
times: error states are reachable if and only if there is a natural
number~$k \in \N$ so that the sentence
\begin{equation}
  \label{eq:errorTrace}
  B[\bar x_0] ~\wedge~ \tr[\bar x_0, \bar x_1] \wedge \tr[\bar x_1, \bar
  x_2]
  \wedge \cdots \wedge
  \tr[\bar x_{k-1}, \bar x_k] ~\wedge~
  \big(E[\bar x_0] \vee E[\bar x_1]
  \vee \cdots \vee E[\bar x_k]\big)
\end{equation}
is $T$-satisfiable. As it turns out, applying SAT or SMT solvers to
\eqref{eq:errorTrace} is a very effective way to discover systems
in which errors can occur, thanks to the performance of such solvers.
If the formula~\eqref{eq:errorTrace} is $T$-unsatisfiable, in
contrast, it is not (at least not in the general case) possible to tell whether
the system can indeed not encounter any errors, or whether the chosen
$k$ was simply too small to find an error.

\subsection{IMC: Interpolation and SAT-Based Model Checking}
\label{sec:imc}

\contents{
\item Focus on propositional interpolation
\item Inference of inductive invariants for transition systems
\item Soundness and completeness
\item Variants/refinements of the algorithm
}

\subsubsection{The IMC Algorithm}

The first interpolation-based model checking algorithm extends BMC
with the computation of inductive invariants through Craig
interpolation, making it possible to tell conclusively that
\eqref{eq:errorTrace} is $T$-unsatisfiable \emph{for every} value of
$k \in \N$.  We refer to the algorithm as the \emph{IMC} algorithm,
following the name given in \cite{DBLP:journals/jar/BeyerLW25}.
To prove that the formula~\eqref{eq:errorTrace} is unsatisfiable for
every $k \in \N$, it is enough to find a sentence~$R[\bar x]$, the
\emph{inductive invariant,} with the properties discussed in
\secref{sec:motivating}:
\begin{itemize}
\item \makebox[6em][l]{\emph{Initiation:}} $B[\bar x] \models_T R[\bar x]$;
\item \makebox[6em][l]{\emph{Consecution:}} $R[\bar x] \wedge \tr[\bar x, \bar x'] \models_T
  R[\bar x']$;
\item \makebox[6em][l]{\emph{Sufficiency:}} $R[\bar x] \wedge E[\bar x] \models_T \bot$.
\end{itemize}

Indeed, assuming that formula~\eqref{eq:errorTrace} holds, due to
\emph{Initiation} and \emph{Consecution} it follows that $R$ is true
for all states~$\bar x_0, \ldots, \bar x_i$, and due to
\emph{Sufficiency} a contradiction with the disjunction of
$E$-formulas follows.

For reasons of presentation, we make the simplifying assumption that
the relation~$\tr$ is \emph{total}, i.e.,
$\forall \bar y.\, \exists \bar z.\; \tr[\bar y, \bar z]$ is
valid.
To obtain candidates for the inductive invariant~$R[\bar x]$, the IMC
algorithm constructs a sequence~$R_0, R_1, R_2, \ldots$ of sentences,
in such a way that each $R_i$ is an over-approximation of the states
that are reachable in at most $i$ steps from the initial
states. Naturally, we can set $R_0 = B$. To construct the other
sentences, we fix constants~$j, k \in \N$ with $j \leq k$, which are
the parameters of IMC. Each formula~$R_{i+1}$ is then derived from
$R_i$ by considering the following conjunction:
\begin{equation}
  \label{eq:imcQuery}
  \underbrace{\vphantom{\bigvee_{l=j}^{k}}
    R_i[\bar x_0] \wedge \tr[\bar x_0, \bar x_1]}_{A[\bar x_0, \bar x_1]}
  ~\wedge~
  \underbrace{
  \bigwedge_{l=1}^{k-1} \tr[\bar x_l, \bar x_{l+1}]
  \wedge
  \bigvee_{l=j}^{k}  E[\bar x_l]}_{B[\bar x_1, \ldots, \bar x_k]}~.
\end{equation}
Formula~\eqref{eq:imcQuery} corresponds to the BMC
query~\eqref{eq:errorTrace}, with the difference that the initial
states~$B$ have been replaced with the $i$th approximation~$R_i$, and that
the range of the disjunction of the $E[\bar x_l]$ is parameterized in $j$.
Three cases are possible:
\begin{itemize}
\item \textbf{Case (i):} Formula~\eqref{eq:imcQuery} is satisfiable
  and $i = 0$. This implies that $R_i = R_0 = B$ and error states are
  reachable in the transition system. Formula~\eqref{eq:errorTrace}
  is then satisfiable as well for the same~$k$.
\item \textbf{Case (ii):} Formula~\eqref{eq:imcQuery} is satisfiable
  and $i > 0$. This can happen either because error states are
  reachable (like in (i)), or because the approximation~$R_i$ is too
  imprecise. Nothing can be said in this situation and the algorithm
  has to abort.
\item \textbf{Case (iii):} Formula~\eqref{eq:imcQuery} is
  unsatisfiable. We can then extract a $T$-separator~$I_i[\bar x_1]$
  from the
  conjunction~$A[\bar x_0, \bar x_1] \wedge B[\bar x_1, \ldots, \bar
  x_k]$, since $\bar x_1$ are the constants that occur in both
  $A[\bar x_0, \bar x_1]$ and $B[\bar x_1, \ldots, \bar x_k]$. By
  renaming $\bar x_1$ to $\bar x$, we obtain a sentence~$I_i[\bar
  x]$ that is an over-approximation of the states that can be reached
  in one step from $R_i$. The next approximation is set to
  $R_{i+1}[\bar x] = R_i[\bar x] \vee I_i[\bar x]$, weakening $R_i$ to
  take one further step of the transition system into
  account.
\end{itemize}

It turns out that the sequence~$R_0, R_1, R_2, \ldots$ often contains
some $R_i$, for small values of $i$, that are inductive
invariants. Intuitively, this is because the sentences~$R_i$ are
constructed to approximate the reachable states of the transition
system. The use of Craig interpolation has the effect of
\emph{accelerating} the fixed-point construction, focusing on
properties of the states that are sufficient to stay inconsistent with
the $B$-part of \eqref{eq:imcQuery} (i.e., not reaching any error
states within $k$ steps) and ignoring other details that
might be part of the formulas~$B, \tr, E$ but irrelevant for ruling
out errors. This effect of acceleration is achieved by obtaining interpolants
from proofs, utilizing the fact that a proof is likely to only refer
to those facts and constraints that are necessary to show the
unsatisfiability of \eqref{eq:imcQuery}, which are exactly the
properties needed to avoid error states.

By construction, the sentences~$R_0, R_1, R_2, \ldots$ all satisfy the
\emph{Initiation} condition of inductive invariants. To check whether
any of the sentences~$R_i$ satisfies \emph{Consecution}, this
condition can be verified explicitly
($R_i[\bar x] \wedge \tr[\bar x, \bar x'] \models_T R_i[\bar x']$);
alternatively, as a sufficient condition, it can also be checked
whether it is ever the case that $R_{i+1} \equiv_T R_i$, which implies
\emph{Consecution} thanks to the definition of
$T$-separators. \emph{Sufficiency} of a sentence~$R_i$ can similarly
be verified explicitly
($R_i[\bar x] \wedge E[\bar x] \models_T \bot$).
IMC can also be set up in such a way that all constructed
formulas~$R_i$ satisfy \emph{Sufficiency}: for this, it has to be
verified that none of the initial states~$R_0 = B$ are error states,
and $j$ must have value~$1$~\cite{DBLP:conf/cav/McMillan03}.
(At this point, the assumption that $T$ is total is required as well.)

To make the approach effective, it is common to require the
formulas~$B, \tr, E$ to be quantifier-free and the considered
theory~$T$ to admit plain quantifier-free interpolation; this ensures
that also the sentences~$R_0, R_1, R_2, \ldots$ are quantifier-free
and the required checks remain decidable.  The original IMC
algorithm~\cite{DBLP:conf/cav/McMillan03} made use of Craig
interpolants extracted from propositional resolution proofs (see
\refchapter{chapter:propositional}).

\bigskip
The parameters~$j, k$ can be used to control the behavior of the
algorithm. Bigger values of $k$ and smaller values of $j$ lead to more
precise over-approximations, since the $B$-part of \eqref{eq:imcQuery}
ensures that no error states are reachable from the
separator~$I_i[\bar x_1]$ in at least $j-1$ and at most $k-1$
steps. Increasing $k$, in particular, will lead to stronger
formulas~$I_i[\bar x_1]$ being computed, which will reduce the
probability of ending up in Case~(ii) in later iterations. It can
be shown that IMC is even \emph{complete} for computing
inductive invariants under certain
assumptions~\cite{DBLP:conf/cav/McMillan03}:
\begin{itemize}
\item The transition system has \emph{bounded reverse depth,} which
  means that there is a number~$d \in \N$ such that whenever $E$ is
  reachable from some state, then there is also a path from the state
  to $E$ that has length at most $d$.
\item The sequence of computed invariant candidates~$R_0, R_1, \ldots$
  is guaranteed to become stationary, i.e., there is some $i \in \N$
  such that $R_i \equiv_T R_{i'}$ for every $i' \geq i$.
\end{itemize}
If the first assumption is satisfied, it is enough to choose a large
enough $k$ and set $j$ to $1$ to prevent the algorithm from failing
due to Case~(ii). If, moreover, the second assumption is satisfied,
there will be some $R_i$ such that $R_i \equiv_T R_{i+1}$, which is
then guaranteed to pass the \emph{Consecution} condition and give rise
to an inductive invariant.

Both assumptions are satisfied by transition systems and invariants in
propositional logic. For the second condition, observe that each $R_i$
entails the next $R_{i+1}$, so that the candidates~$R_0, R_1, \ldots$
form an ascending chain. Since there are only finitely many
propositional formulas~$R_i$ (up to equivalence), every
chain~$R_0, R_1, \ldots$ eventually becomes stationary. In contrast,
for systems with infinite state space, which are modeled using
theories~$T$ that allow infinite universes, the conditions are
typically violated. Empirically, the IMC algorithm tends to perform
well even for such systems~\cite{DBLP:journals/jar/BeyerLW25}.

\subsubsection{Related Interpolation-Based Algorithms}

Many different versions and improvements of the SAT-based model
checking algorithm IMC have been proposed. In the original
formulation, IMC was defined for transition systems and interpolation
in propositional logic~\cite{DBLP:conf/cav/McMillan03}. IMC was later
extended to theories and the setting of software
verification~\cite{DBLP:journals/jar/BeyerLW25}, a setting that we
adopted for this chapter as well.

The interpolation sequence-based model checking algorithm \emph{ISB} uses a
similar fixed-point loop as IMC, but applies sequence interpolants
(\defref{def:sequenceInterpolants}) instead of binary
interpolants~\cite{DBLP:conf/fmcad/VizelG09}. The key difference
between IMC and ISB is in the way how interpolation queries are
defined: while IMC includes the previous approximation~$R_i$ in the
query~\eqref{eq:imcQuery} used to compute the next $R_{i+1}$, the
queries applied by ISB only refer to the sentences~$B, \tr, E$ defining
the transition system and the error states, but ISB increases the
considered number of transitions (copies of $\tr$) in each
iteration.

Verification algorithms for software programs often take the control
structure of a program into account. Instead of a monolithic
transition system~$(B, \tr)$, we would then start from a program
described by a set of transition relations arranged by a control-flow
automaton. This setting leads to smaller interpolation queries, since
only one program path at a time has to be considered.  The \emph{Impact}
algorithm follows this strategy, using sequence interpolants to obtain
invariant candidates for particular program
paths~\cite{DBLP:conf/cav/McMillan06}.

The \emph{IC3} algorithm is even more local and analyzes individual
system states instead of the whole transition relation to obtain
invariant candidates~\cite{DBLP:conf/vmcai/Bradley11}. While IC3 has
similarities with the IMC and ISB algorithms, IC3 does not rely on
Craig interpolation anymore but directly infers clauses that are
inductive invariants. The \emph{Spacer} algorithm, which is an
extension of IC3, uses interpolation as an optional abstraction
step~\cite{DBLP:journals/fmsd/KomuravelliGC16}.


\subsection{CEGAR: Counterexample-Guided Abstraction Refinement}
\label{sec:cegar}

\contents{
\item Focus on interpolation in theories
\item Inference of inductive state invariants
\item Boolean and Cartesian abstractions
\item Soundness, completeness, incompleteness
}

Predicate
abstraction~\cite{DBLP:conf/cav/GrafS97,DBLP:conf/tacas/BallPR01} and
\emph{Counterexample-Guided Abstraction Refinement (CEGAR)}~\cite{cegar}
provide another perspective on the use of Craig interpolation for
verifying systems. Both notions predate the IMC algorithm and the use
of interpolation in verification but turned out to be very
suitable for a combination with
interpolation~\cite{DBLP:conf/popl/HenzingerJMM04,DBLP:conf/tacas/McMillan05}.

\subsubsection{Existential Abstractions}

We again fix a theory~$T$. For the sake of presentation, we assume that
the theory comes with a uniquely defined structure~$S_T$, i.e., that
the set~$\mathcal{S}_T = \{S_T\}$ in \defref{def:theory} is a singleton
set. As the states of a system are described by a
vector~$\bar x = (x_1, \ldots, x_n)$ of constants with
sorts~$\bar \sigma = (\sigma_1, \ldots, \sigma_n)$, the state space of
the system is the Cartesian
product~$D = D_{\sigma_1} \times \cdots \times D_{\sigma_n}$. Since
the theory has a unique structure~$S_T$, with slight abuse of notation
we can write $B[s]$ to express that a state~$s \in D$ satisfies the
sentence~$B$, i.e., that $s$ is an initial state, and similarly state
by $\tr[s, s']$ that there is a transition from $s$ to $s'$.

To verify properties of the transition system~$(B, \tr)$, for instance
the reachability of error states~$E$, it is frequently useful to
consider an \emph{abstract} version of the transition system with a
smaller number of states. We can define abstractions by choosing a
non-empty set~$\hat D$ as the abstract state space, together with an
\emph{abstraction function}~$h : D \to \hat D$ that induces the
abstract transition system in terms of a set~$\hat B$ of abstract
initial states and an abstract transition relation~$\hat \tr$:
\begin{align*}
  \hat B &= \{ \hat s \in \hat D \mid
  \text{there is~} s \in D \text{~with~} h(s) = \hat s \text{~and~}
           B[s] \}~,
  \\
  \hat \tr &= \{ (\hat s, \hat s') \in \hat D^2 \mid
  \text{there are~} s, s' \in D \text{~with~} h(s) = \hat s, h(s') =
  \hat s' \text{~and~} \tr[s, s'] \}~.
\end{align*}
This construction is known as \emph{existential
  abstraction}~\cite{DBLP:conf/fmco/Grumberg05}. It is useful since
non-reachability of states in existential abstractions implies
non-reachability also in the concrete system: whenever $M \subseteq D$
and the abstract states~$h(M)$ are not reachable from initial states
in $(\hat B, \hat\tr)$, the concrete states $M$ are not reachable
from the initial states in $(B, \tr)$ either. More generally,
existential abstractions can be used to prove ACTL* properties of the
concrete system~\cite{DBLP:conf/fmco/Grumberg05}.

Existential abstractions are illustrated in
\figref{fig:existentialAbs1}. Black states and edges represent the
concrete transition system~$(B, \tr)$, with initial states marked with
a further incoming edge. The concrete states are grouped together to
\ifcolor red\else gray\fi{} abstract states by the
function~$h$. Abstract states are connected by transitions whenever
any of the included concrete states are, leading to over-approximation
of the behavior. For instance, the concrete transition system in
\figref{fig:existentialAbs1} does not contain any transitions of a
state to itself, while the abstract transition system has such
transitions.

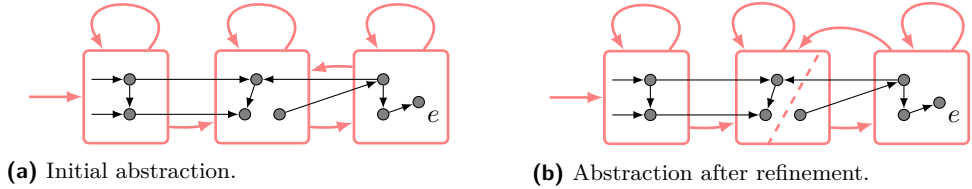
\begin{figure}[tb]
  \tikzset{dot/.style={circle,draw=black,fill=\ifcolor black!50\else black\fi,inner sep=1.5pt},
    concEdge/.style={-latex},
    absBase/.style={draw=\ifcolor red!50\else black!30\fi,line width=1pt},
    absEdge/.style={-latex,absBase},
    absNode/.style={rectangle, rounded corners=2pt,absBase, inner sep=2ex}}

  \hspace*{\fill}
  \begin{minipage}{0.45\linewidth}
    \begin{subfigure}{\linewidth}
      \centering
  \begin{tikzpicture}[node distance=2ex]
    \node[dot] (a) {};
    \node[dot,below=of a] (b) {};
    \node[dot,right=1.5 of a] (c) {};
    \node[dot,below=of c,xshift=-1ex] (d) {};
    \node[dot,below= of c,xshift=2ex] (e) {};
    \node[dot,right=1.5 of c] (f) {};
    \node[dot,below=of f] (g) {};
    \node[dot,right=of f,yshift=-2ex] (h) {};

    \node[absNode,fit={(a) (b) ($(a)+(-0.3,0)$) ($(a)+(0.2,0)$)}] (A) {};
    \node[absNode,fit={(c) (d) (e)}] (B) {};
    \node[absNode,fit={(f) (g) (h)}] (C) {};

    \node[below right=-0.1 of h] {$e$};
    
    \draw[concEdge] ($(a)+(-0.5,0)$) -- (a);
    \draw[concEdge] ($(b)+(-0.5,0)$) -- (b);

    \draw[concEdge] (a) -- (b);
    \draw[concEdge] (a) -- (c);
    \draw[concEdge] (b) -- (d);
    \draw[concEdge] (c) -- (d);
    \draw[concEdge] (e) -- (f);
    \draw[concEdge] (f) -- (c);
    \draw[concEdge] (f) -- (g);
    \draw[concEdge] (g) -- (h);
    
    \draw[absEdge] ($(A.west)+(-0.7,0)$) -- (A);
    \draw[absEdge] ($(A.east)+(0,-0.4)$) to [bend right=10] (B);
    \draw[absEdge] ($(B.east)+(0,-0.4)$) to [bend right=10] (C);
    \draw[absEdge] ($(C.west)+(0,0.4)$) to [bend right=10] (B);
    \draw[absEdge, loop above,out=50,in=120,looseness=4]
         ($(A.north)+(0.3,0)$) to ($(A.north)+(-0.3,0)$);
    \draw[absEdge, loop above,out=50,in=120,looseness=4]
         ($(B.north)+(0.3,0)$) to ($(B.north)+(-0.3,0)$);
    \draw[absEdge, loop above,out=50,in=120,looseness=4]
         ($(C.north)+(0.3,0)$) to ($(C.north)+(-0.3,0)$);
  \end{tikzpicture}
  \caption{Initial abstraction.}
  \label{fig:existentialAbs1}
  \end{subfigure}
  \end{minipage}
  \hspace*{\fill}
  \begin{minipage}{0.45\linewidth}
  \begin{subfigure}{\linewidth}
      \centering
  \begin{tikzpicture}[node distance=2ex]
    \node[dot] (a) {};
    \node[dot,below=of a] (b) {};
    \node[dot,right=1.5 of a] (c) {};
    \node[dot,below=of c,xshift=-1ex] (d) {};
    \node[dot,below= of c,xshift=2ex] (e) {};
    \node[dot,right=1.5 of c] (f) {};
    \node[dot,below=of f] (g) {};
    \node[dot,right=of f,yshift=-2ex] (h) {};

    \node[absNode,fit={(a) (b) ($(a)+(-0.3,0)$) ($(a)+(0.2,0)$)}] (A) {};
    \node[absNode,fit={(c) (d) (e)}] (B) {};
    \node[absNode,fit={(f) (g) (h)}] (C) {};

    \node[below right=-0.1 of h] {$e$};
    
    \draw[absBase,dashed] ($(B.south)+(-0.2,0)$) -- ($(B.north)+(0.5,0)$); 

    \draw[concEdge] ($(a)+(-0.5,0)$) -- (a);
    \draw[concEdge] ($(b)+(-0.5,0)$) -- (b);

    \draw[concEdge] (a) -- (b);
    \draw[concEdge] (a) -- (c);
    \draw[concEdge] (b) -- (d);
    \draw[concEdge] (c) -- (d);
    \draw[concEdge] (e) -- (f);
    \draw[concEdge] (f) -- (c);
    \draw[concEdge] (f) -- (g);
    \draw[concEdge] (g) -- (h);
    
    \draw[absEdge] ($(A.west)+(-0.7,0)$) -- (A);
    \draw[absEdge] ($(A.east)+(0,-0.4)$) to [bend right=10] (B);
    \draw[absEdge] ($(B.east)+(0,-0.4)$) to [bend right=10] (C);
    \draw[absEdge] ($(C.north)+(-0.4,0)$) to [bend right=50] ($(B.north)+(0.2,0)$);
    \draw[absEdge, loop above,out=50,in=120,looseness=4]
         ($(A.north)+(0.3,0)$) to ($(A.north)+(-0.3,0)$);
    \draw[absEdge, loop above,out=50,in=120,looseness=5]
         ($(B.north)+(0,0)$) to ($(B.north)+(-0.5,0)$);
    \draw[absEdge, loop above,out=50,in=120,looseness=5]
         ($(C.north)+(0.4,0)$) to ($(C.north)+(-0.1,0)$);

  \end{tikzpicture}
  \caption{Abstraction after refinement.}
  \label{fig:existentialAbs2}
  \end{subfigure}
  \end{minipage}
  \hspace*{\fill}
  
  \caption{A transition system (black) and two existential abstractions
    (\ifcolor red\else gray\fi).}
  \label{fig:existentialAbs}
\end{figure}

As a result, the reachability of states in $(\hat B, \hat\tr)$ does not
imply the reachability of the corresponding states in $(B, \tr)$,
since the abstract transition system can contain \emph{spurious} paths
that have no equivalent paths in $(B, \tr)$. For instance, in
\figref{fig:existentialAbs1} there is an abstract path from the left-most
initial to the right-most state, but there is no concrete path from an
initial state to the right-most state~$e$.

\subsubsection{Predicate Abstraction}

A convenient way to define existential abstractions is \emph{predicate
  abstraction}~\cite{DBLP:conf/cav/GrafS97}. For this, we assume a
finite set~$\mathcal{P} = \{P_1[\bar x], \ldots, P_k[\bar x]\}$ of
sentences that can refer to the constants~$\bar x$ representing system
states. The transition system defined using \emph{Boolean
  abstraction}~\cite{DBLP:conf/tacas/BallPR01}, one of multiple
different versions of predicate abstraction, groups together states
that cannot be distinguished by the predicates in
$\mathcal{P}$. Abstract states can be represented as bit-vectors,
$\hat D = \{0, 1\}^k$, storing for each of the
predicates in $\mathcal{P}$ whether the predicate holds for some
concrete state~$s \in D$ or not:
\begin{equation*}
  h(s) = (p_1(s), \ldots, p_k(s)),
  \qquad
  p_i(s) =
  \begin{cases}
    1 & \text{if~} P_i[s]
    \\
    0 & \text{otherwise}
  \end{cases}
  \qquad
  \text{for~} i \in \{1, \ldots, k\}~.
\end{equation*}

For an abstraction with $k$ predicates, this results in a transition
system~$(\hat B, \hat \tr)$ with $2^k$ abstract states, which can also
be constructed explicitly in terms of the sentences~$B, \tr$ defining
the concrete transition system. For this, we define
$p^{-1}(\hat s)[\bar x]$ as the conjunction of the (positive or
negative) predicates corresponding to an abstract
state~$\hat s = (b_1, \ldots, b_k) \in \hat D$:
\begin{align*}
  p^{-1}(\hat s)[\bar x] &~=~
                           \bigwedge_{i=1}^k (P_i[\bar x] \leftrightarrow b_i)~.
  \\
  \intertext{The abstract transition system is then given by:}
  \hat B &~=~ \{ \hat s \in \hat D \mid B[\bar x] \wedge p^{-1}(\hat s)[\bar
           x]
           \text{~is $T$-satisfiable}\}~,
  \\
  \hat \tr &~=~ \{ (\hat s, \hat s') \in \hat D^2 \mid
           p^{-1}(\hat s)[\bar
           x] \wedge   \tr[\bar x, \bar x'] \wedge p^{-1}(\hat s')[\bar
           x']
           \text{~is $T$-satisfiable}\}~.
\end{align*}

\subsubsection{Counterexample-Guided Abstraction Refinement}

The key idea of the CEGAR algorithm is to \emph{refine} abstractions
when they are too coarse to prove some desired property of the
concrete transition system; this way, the verification of some system
can start with a very small, imprecise abstract transition system
(e.g., the system obtained for an empty
set~$\mathcal{P}_0 = \emptyset$ of predicates) and gradually refine
the abstraction until it is possible to prove the
property~\cite{cegar}. Refinement can be seen as sub-dividing the
states of an abstract transition system, so that spurious paths are
eliminated. For instance, in \figref{fig:existentialAbs2}, the abstract
state in the middle has been split into two states, this way removing
abstract paths and making the abstract state containing $e$
unreachable from the initial state. In general, multiple abstract
states have to be split to remove paths in the abstract transition
system. Abstraction refinement can be achieved with the help of Craig
interpolation~\cite{DBLP:conf/popl/HenzingerJMM04}.

For this, suppose $E[\bar x]$ is a sentence representing error states,
and $\hat E = \{ \hat s \in \hat D \mid
  \text{there is~} s \in D \text{~with~} h(s) = \hat s \text{~and~}
           E[s] \}$ is the corresponding abstract
set. Suppose that abstract analysis with the
predicates~$\mathcal{P}_l$ has identified a
path~$\hat s_1, \hat s_2, \ldots, \hat s_m$ of states such that
$\hat s_1 \in \hat B$, $(\hat s_i, \hat s_{i+1}) \in \hat \tr$ for all
$i \in \{1, \ldots, m-1\}$, and $\hat s_m \in \hat E$, but no
corresponding path exists in the concrete transition system. The
latter assumptions means that the formula
\begin{equation*}
  \underbrace{B[\bar x_1]}_{A_1} \wedge
  \underbrace{\tr[\bar x_1, \bar x_2]}_{A_2} \wedge
  \underbrace{\tr[\bar x_2, \bar x_3]}_{A_3} \wedge \cdots \wedge 
  \underbrace{\tr[\bar x_{m-1}, \bar x_m]}_{A_m} \wedge
  \underbrace{E[\bar x_m]}_{A_{m+1}}
\end{equation*}
is $T$-unsatisfiable. The path~$\hat s_0, \hat s_1, \ldots, \hat s_m$
is also called a ``spurious counterexample'' that can be eliminated by
adding further predicates, increasing the precision of predicate
abstraction. Let $I_0, I_1, \ldots, I_{m+1}$ be a sequence interpolant
for the conjunction~$A_1 \wedge \cdots \wedge A_{m+1}$. By definition,
each $I_i$ for $i \in \{1, \ldots, m\}$ only refers to the
symbols~$\bar x_i$, so that we obtain $m$ new
sentences~$I_1[\bar x], \ldots, I_m[\bar x]$ by renaming; those
sentences give rise to the refined abstract transition system, defined
by
$\mathcal{P}_{l+1} = \mathcal{P}_{l} \cup \{I_1[\bar x], \ldots,
I_m[\bar x]\}$.

In the refined abstract system, the spurious counterexample does not
exist any more, in the sense that there is no abstract
path~$\hat s'_1, \hat s'_2, \ldots, \hat s'_m$ of length~$m$ to an
error state. This is due to the new predicates: in every initial
state~$\hat s \in \hat B$, the bit corresponding to the new
predicate~$I_1$ has to be $1$; therefore, in every state reachable in
one step from an initial state, the bit for $I_2$ has to be $1$,
etc. Finally, the predicate~$I_m$, which is present in every state
reachable in $m-1$ steps, is inconsistent with the sentence~$E$
defining error states, implying that no paths of length~$m$ to error
states are possible. There might be, of course, paths of different
length to error states, which can be removed by further refinement
steps.

\subsubsection{Variants of Predicate Abstraction and CEGAR}

The explicit construction of abstract transition systems tends to be a
time-consuming step and is often avoided in newer implementations
of CEGAR using the \emph{lazy abstraction}
paradigm~\cite{DBLP:conf/popl/HenzingerJMS02}. Instead of working
explicitly with abstract states, lazy abstraction constructs a
\emph{reachability graph} in which every node is labeled with a
Boolean combination of predicates in $\mathcal{P}$, summarizing a
whole set of abstract states. An edge from a node~$n$ to $n'$ implies
that the formulas~$\phi, \phi'$ labeling $n, n'$, respectively,
satisfy
$\phi[\bar x] \wedge \tr[\bar x, \bar x'] \models_T \phi'[\bar
x']$. The reachability graph is built on-the-fly during exploration,
starting from the initial states, and refined when error states are
encountered.  The Impact algorithm does not even store and
reuse predicates, but instead generates new predicates (using Craig
interpolation) each time a new path is
considered~\cite{DBLP:conf/cav/McMillan06}.

A further optimization, utilized in most papers on CEGAR, is to take
again the control structure of software programs into account (e.g.,
\cite{DBLP:conf/popl/HenzingerJMS02,
  DBLP:conf/cav/McMillan06,DBLP:conf/popl/HeizmannHP10}). Besides
leading to smaller interpolation queries, since
only one path is considered at a time, sensitivity to control flow enables
verification tools to maintain a separate set of predicates for each
control location. The even more general framework of constrained Horn
clauses, described in the next section, can handle many different
kinds of control structure, including procedure calls and concurrency.


\subsection{CHC: Craig Interpolation and Constrained Horn Clauses}
\label{sec:chc}

\contents{
\item Rephrasing the verification problem as the problem of satisfying
  Horn constraints
\item Relationship between recursion-free Horn clauses and Craig interpolation
\item Inference of function contracts and other annotations
\item Is there a connection to Databases? Maybe not, only that Horn
  clauses (Datalog) are used in both cases. But in the database case,
  interpolants in the Horn fragment are computed.
}

The framework of \emph{Constrained Horn Clauses} (CHCs) rephrases
program verification in terms of model construction in an extension of
a theory~$T$ with additional
predicates~\cite{DBLP:conf/pldi/GrebenshchikovLPR12,DBLP:conf/birthday/BjornerGMR15}. CHCs
can be seen as an intermediate language that lets verification
front-ends communicate in a standardized way with verification
back-ends. CHCs also subsume the extended versions of Craig
interpolation (parallel interpolants, sequence interpolants, tree
interpolants) that were defined in \secref{sec:extendedCraig} and let us
state a corresponding, more general result.

\subsubsection{Definitions and Basic Properties}

\begin{definition}[Constrained Horn Clause~\cite{DBLP:journals/fmsd/RummerHK15}]
  Suppose $T$ is a theory and $\mathcal{R}$ a set of relation symbols
  disjoint from the signature~$\Sigma_T$ of $T$. A \emph{Constrained
    Horn Clause (CHC)} is a
  sentence~$\forall \bar x.\; C \wedge B_1 \wedge \cdots \wedge B_n
  \to H$ in which:
  \begin{itemize}
  \item $C$ is a formula over $\Sigma_T$ that can refer to the
    variables~$\bar x$;
  \item each $B_i$ is an application~$p(t_1, \ldots, t_k)$ of a
    relation symbol~$p \in \cal R$ to terms over
    $\Sigma_T$ and $\bar x$;
  \item $H$ is either $\bot$ or  an
    application~$p(t_1, \ldots, t_k)$ of $p \in \cal R$ to terms over
    $\Sigma_T$ and $\bar x$.
  \end{itemize}
\end{definition}

We call $C$ the \emph{constraint,} $H$ the \emph{head,} and
$C \wedge B_1 \wedge \cdots \wedge B_n$ the \emph{body} of the
clause. Constraints~$C = \top$ are usually left out. When modeling the
partial correctness of a program, the relation symbols~$\cal R$ are
used to represent verification artifacts like inductive invariants or
contracts, while the clauses model conditions those artifacts have to
satisfy. Verification is carried out by \emph{constructing a model} of
the clauses, which then includes all the necessary artifacts
witnessing the correctness of the program. The algorithms we have
encountered in \secref{sec:imc} and \secref{sec:cegar} can be adapted to
work on sets of CHCs.

This process is formalized using two notions of solvability of clauses.  For
the sake of presentation, we focus again on theories~$T$ that come with a
uniquely defined structure~$S_T$, i.e., for which the
set~$\mathcal{S}_T = \{S_T\}$ in \defref{def:theory} is a singleton
set.  A set of CHCs is called \emph{semantically solvable} if the
structure~$S_T$ can be extended to a structure over the
signature~$\Sigma_T \cup \cal R$ that satisfies all clauses, i.e., if
the set of CHCs is $T$-satisfiable according to
\defref{def:theory}.\footnote{\label{fn:looseSemantics}For the general case of theories~$T$
  with multiple structures, two definitions of semantic solvability
  are commonly used. In \emph{loose
    semantics}~\cite{DBLP:journals/fmsd/RummerHK15,DBLP:conf/birthday/BjornerGMR15},
  a set of CHCs is considered semantically solvable if \emph{every} $T$-structure
  can be extended to a structure satisfying all clauses. Existing CHC
  solvers more commonly apply the \emph{strict semantics} also used by
  SMT solvers, which asks whether there exists \emph{some}
  $T$-structure that can be extended to a structure satisfying all
  clauses. CHCs are often formulated in such a way that both notions
  of semantics lead to the same result.} A set of CHCs is called
\emph{syntactically solvable} if all clauses can be made
$T$-valid by uniformly substituting formulas over the
signature~$\Sigma_T$ for the relation symbols from
$\cal R$~\cite{DBLP:journals/fmsd/RummerHK15}.\todo{should we define
  this notion formally?} Syntactic solvability is related to the
classical \emph{solution problem} in
logic~\cite{DBLP:conf/frocos/Wernhard17}.

\begin{example}
  We model the Fibonacci example from \figref{fig:fibonacci} using
  CHCs. For this, we introduce a relation symbol $\mathit{loop}$ that
  represents the loop invariant of the function~\texttt{fib} and
  formulate clauses for the \emph{Initiation,} \emph{Consecution,} and
  \emph{Sufficiency} conditions. Clause~\eqref{clause:loop1} tells
  that the invariant should hold when entering the loop, i.e., for any
  value~$n$ with $n \geq 0$ and $a = 0, b = 1, i =
  0$. Clause~\eqref{clause:loop2} ensures that the invariant is
  preserved by the loop body: whenever $\mathit{loop}$ holds for
  values~$n,a,b,i$ and the loop condition~$i < n$ is satisfied,
  $\mathit{loop}$ should also hold for the state after one loop
  iteration. Clause~\eqref{clause:loop3} ensures that whenever
  $\mathit{loop}$ holds for values~$n,a,b,i$ and the loop terminates
  due to $i \geq n$, the postcondition~$a \geq 0$ is not violated.
  \begin{align}
    %
    \forall n.\; n \geq 0 &\to \mathit{loop}(n, 0, 1, 0) \label{clause:loop1}
    \\
    \forall n, a, b, i.\; i < n \wedge \mathit{loop}(n, a, b, i) &\to \mathit{loop}(n, b,
                                                                   a+b, i+1) \label{clause:loop2}
    \\
    \forall n, a, b, i.\; i \geq n \wedge a \not\geq 0 \wedge \mathit{loop}(n, a, b, i)
                          &\to \bot \label{clause:loop3}
  \end{align}
  Every syntactic solution of the clauses corresponds to an inductive
  loop invariant. The existing CHC solvers, which are often
  interpolation-based, can easily compute such solutions; one possible
  solution is the assignment
  $\{\mathit{loop}(n, a, b, i) \mapsto a \geq 0 \wedge b \geq 0 \}$,
  proving that the program is partially correct. Note that $a \geq 0
  \wedge b \geq 0$ coincides with the inductive invariant that was
  discovered in \secref{sec:motivating}.
  \lipicsEnd
\end{example}

If a set of CHCs is syntactically solvable, then it is also
semantically solvable. The converse is not true in general, because a
set of CHCs could only have models that cannot be defined explicitly
in the theory~$T$~\cite{DBLP:journals/fmsd/RummerHK15}. For the
fragment of \emph{recursion-free} CHCs, however, the notions of
semantic and syntactic solvability coincide thanks to the Craig
interpolation theorem.
To illustrate the relationship between CHCs and $T$-separators,
consider formulas~$A[\bar x, \bar y], B[\bar x, \bar z]$ that only
have the variables~$\bar x$ and symbols from $\Sigma_T$ in common,
while $\bar y, \bar z$ are local variables that only occur in
$A[\bar x, \bar y]$ or $B[\bar x, \bar z]$, respectively. Consider
then the CHCs
$\forall \bar x, \bar y.\; A[\bar x, \bar y] \to I(\bar x)$ and
$\forall \bar x, \bar z.\; B[\bar x, \bar z] \wedge I(\bar x) \to
\bot$. By definition, any syntactic solution of the two clauses is a
$T$-separator\footnote{We consider the free
  variables~$\bar x, \bar y, \bar z$ as constant symbols at this
  point.} for the
conjunction~$A[\bar x, \bar y] \wedge B[\bar x, \bar z]$, and vice
versa, implying that the two CHCs are syntactically solvability if and
only if $A[\bar x, \bar y] \wedge B[\bar x, \bar z]$ is
$T$-unsatisfiable. In a similar way, sequence interpolants and tree
interpolants can be represented using sets of CHCs.

To state this result more generally, for a set~$\cal C$ of CHCs over
the relation symbols~$\cal R$ we introduce a \emph{dependence
  relation}~$\to_{\mathcal{C}}\; \subseteq {\cal R}^2$, defining
$p \to_{\cal C} q$ if there is a CHC in $\cal C$ that contains $p$ in
the head and $q$ in the body. The set~$\cal C$ is called
\emph{recursion-free} if $\to_{\cal C}$ is acyclic, and
\emph{recursive} otherwise.

\begin{theorem}[Solvability of
  CHCs~\cite{DBLP:journals/fmsd/RummerHK15}]
  \label{thm:chcRecursionfree}
  Let $\cal C$ be a finite and recursion-free set of CHCs over a
  theory~$T$. Then $\cal C$ is semantically solvable\footnote{For
    theories with multiple models, applying \emph{loose
  semantics} as discussed in Footnote~\ref{fn:looseSemantics}.} if and only if it is syntactically solvable.
\end{theorem}

To prove the direction~``$\Rightarrow$'' of this result, the
recursion-free CHCs can be expanded to obtain a disjunctive
interpolation problem~\cite{DBLP:conf/cav/RummerHK13}, which can be
translated further to a tree interpolation problem; the syntactic
solution can be extracted from the tree interpolant modulo
$T$. Recursion-free CHCs are exponentially more succinct than tree
interpolation problems, however. The problem of determining whether a
recursion-free set of CHCs over the theory~LIA has solutions is
co-NEXPTIME-complete, while checking whether a formula is
LIA-unsatisfiable (and interpolants exist) is only co-NP-complete.

The relationship between semantic and syntactic solvability also
holds in the quantifier-free setting, provided that the theory~$T$
admits quantifier-free interpolation:
\begin{lemma}[Quantifier-free solvability of CHCs]
  Suppose $T$ is a theory that admits plain quantifier-free
  interpolation. Let $\cal C$ be a finite and recursion-free set of
  CHCs over $T$ in which the constraints do not contain quantifiers.
  Then $\cal C$ is semantically solvable\footnote{As in
    \thmref{thm:chcRecursionfree}, for theories with multiple models,
    applying \emph{loose semantics} as discussed in
    Footnote~\ref{fn:looseSemantics}.} if and only if it is
  syntactically solvable, with a solution that does not contain
  quantifiers.
\end{lemma}

\subsubsection{Encoding of Recursive Programs using CHCs}
\label{sec:recursive}

\begin{wrapfigure}{r}{0.45\linewidth}
\begin{lstlisting}
int fibR(int n) {
  if (n == 0)
    return 0;
  else if (n == 1)
    return 1;
  else
    return fibR(n-1) + fibR(n-2);
}
\end{lstlisting}

\caption{Recursive program computing the $n$th element of the Fibonacci sequence.}
  \label{fig:fibonacciRec}
\vspace*{-0.5ex}
\end{wrapfigure}

\noindent
Constrained Horn clauses can be used also for inference tasks beyond
just inductive invariants, making it possible, among others, to
automatically verify recursive or concurrent
programs~\cite{DBLP:conf/pldi/GrebenshchikovLPR12}. Consider the
program in \figref{fig:fibonacciRec}, which is a recursive
implementation of the Fibonacci function from
\secref{sec:motivating}. Properties of the recursive implementation
can be verified by applying the principle of \emph{Design by
  Contract}~\cite{DBLP:journals/computer/Meyer92}, describing the
behavior of the function using pre- and postconditions that serve a
similar purpose as an inductive loop invariant: every recursive call
to the function has to satisfy the precondition and can then rely on
the postcondition being established by the function upon return.

We illustrate how this style of reasoning can be used to prove a
property about the Fibonacci implementation: whenever~$n \geq 3$, the result
of the function is a number~$\verb!fibR!(n) \geq 2$. There are several
possible ways to formulate CHCs for this recursive
function~\cite{DBLP:conf/birthday/BjornerGMR15}. Choosing an encoding
that is close to the classical Design by Contract methodology, we can
reason about the function by introducing two relation
symbols~$\mathit{pre}$ and $\mathit{post}$, representing the pre- and
postcondition of \inl!fibR!, respectively. The
precondition~$\mathit{pre}(n)$ is a formula about the function
argument~$n$, whereas the postcondition~$\mathit{post}(n, r)$ is a
relation between the argument~$n$ and the result~$r$. We then formulate CHCs
that state the properties those symbols have to satisfy:
\begin{align}
  \forall n.\; n \geq 3 &\to \mathit{pre}(n) \label{clause:rec1}
  \\
  \forall n.\; n = 0 \wedge \mathit{pre}(n) &\to \mathit{post}(n, 0) \label{clause:rec2}
  \\
  \forall n.\; n \not= 0 \wedge n = 1 \wedge \mathit{pre}(n) &\to \mathit{post}(n, 1) \label{clause:rec3}
  \\
  \forall n.\; n \not= 0 \wedge n \not= 1 \wedge \mathit{pre}(n)
                        &\to \mathit{pre}(n - 1) \label{clause:rec4}
  \\
  \forall n.\; n \not= 0 \wedge n \not= 1 \wedge \mathit{pre}(n)
                        &\to \mathit{pre}(n - 2) \label{clause:rec5}
  \\
  \forall n, r, s.\; 
  n\not= 0 \wedge n \not= 1 \wedge \mathit{pre}(n)  \wedge 
  \mathit{post}(n-1, r) \wedge \mathit{post}(n-2, s) 
                        &\to  \mathit{post}(n, r+ s) \label{clause:rec6}
  \\
  \forall n, r.\; n \geq 3 \wedge r \not\geq 2 \wedge \mathit{post}(n, r)  &\to \bot \label{clause:rec7}
\end{align}
Clause~\eqref{clause:rec1} states that the contract to be derived
should include inputs~$n \geq 3$, while \eqref{clause:rec7} expresses
that the result~$r$ for an input~$n \geq 3$ should not be less than
$2$. The other clauses encode the recursive function
itself. Clauses~\eqref{clause:rec2} and \eqref{clause:rec3} model the
\inl!return! statements in lines~3 and 5 by stating that the
postcondition has to hold for the result~$0$ and $1$,
respectively. The third \inl!return! statement in line~7
performs two recursive calls and gives rise to somewhat more
complicated clauses: clauses~\eqref{clause:rec4} and
\eqref{clause:rec5} ensure that the recursive calls satisfy the
precondition of the function, while \eqref{clause:rec6} states that
the computed sum~$r + s$ has to satisfy the postcondition, assuming
that the values~$r, s$ produced by the recursive calls satisfy the
postcondition. A solution of the clauses is given by the
assignment~$\{ \mathit{pre}(n) \mapsto n \geq 0,\break \mathit{post}(n, r)
\mapsto n \geq 0 \wedge r \geq 0 \wedge (r < 2 \to n = 0 \vee (n = 1 \wedge r = 1) \vee (n = 2 \wedge r = 1)) \}$.


The clause~\eqref{clause:rec6} is called a \emph{non-linear}
clause, since its body contains three atoms
$\mathit{pre}(n), \mathit{post}(n-1, r), \mathit{post}(n-2, s)$,
whereas all other clauses are \emph{linear} and have at most one body
atom. Non-linear clauses naturally occur when translating recursive or
concurrent programs to CHCs and are handled by CHC solvers with the
help of \emph{tree interpolation.}

\section{Conclusions}

\contents{
\item Outlook: uniform interpolation
}

The chapter has explored the use of Craig interpolation in program
verification, as well as some of the techniques that have been
developed for effectively computing Craig interpolants in SMT
solvers. Craig interpolation is today a standard method applied in
verification tools and there are many more directions that could be
discussed, but which are beyond the scope of the chapter, including:
\begin{itemize}
\item Techniques for deriving Craig interpolants in \emph{combined
    theories,} for instance in the (signature-disjoint) combination of
  arrays and LRA. As a sufficient criterion for quantifier-free
  interpolation in combined theories, the notion of theories being
  \emph{equality interpolating} has been
  identified~\cite{DBLP:conf/cade/YorshM05}. Equality interpolation
  was later shown to be equivalent to the model-theoretical notion of
  \emph{strong
    sub-amalgamation}~\cite{DBLP:journals/tocl/BruttomessoGR14}.
\item The use of \emph{uniform interpolants} in program verification,
  which are in this context also called
  \emph{covers}~\cite{DBLP:conf/esop/GulwaniM08,DBLP:conf/aiia/Gianola22}. Uniform
  interpolation can be used to define a model checking algorithm that
  is similar in spirit to the IMC algorithm (\secref{sec:imc}), but in
  which the computed invariant candidates~$R_i$ are exact.
\item Methods for \emph{controlling} the result of Craig
  interpolation. In applications like verification, not all
  interpolants work equally well, and the ability of algorithms like
  IMC to compute inductive invariants strongly depends on the right
  interpolants being found by the applied interpolation procedure. To
  address this, methods have been developed to control the logical
  strength~\cite{DBLP:conf/vmcai/DSilvaKPW10,DBLP:conf/atva/GurfinkelRS13},
  size~\cite{DBLP:conf/vstte/AltFHS15}, or the
  vocabulary~\cite{DBLP:conf/esop/DSilva10,DBLP:journals/acta/LerouxRS16}
  of interpolants.
\end{itemize}

\section*{Acknowledgments}
\addcontentsline{toc}{section}{Acknowledgments}

The author would like to thank Jean Christoph Jung, Sibylle
M\"ohle-Rotondi, and Christoph Wernhard for the careful proof-reading
and helpful discussions while preparing this chapter.  The work
on the chapter been supported by the Swedish Research Council (VR)
under grant~2021-06327, the Wallenberg project UPDATE, the Ethereum
Foundation under grant~FY25-2124, and the EuroProofNet COST action.

\bibliography{papers}

\end{document}